%% file: body.tex
\documentclass[3p,times]{elsarticle}

\nopreprintlinetrue

\usepackage{ecrc_modified}

\volume{}%{00}

\firstpage{1}

\journalname{Preprint}%{Preprint submitted to Discrete Applied Mathematics}

\runauth{Arslan, Guralnik, Koditschek}

\jid{}%{DAM}

\jnltitlelogo{}%{DISCRETE \\ APPLIED \\ MATHEMATICS}

\biboptions{sort&compress}

\input{package}

\input{header}

\input{notation}

\begin{document}

\begin{frontmatter}



%% e.g. \dochead{17th International Conference on Dynamical Processes in Excited States of Solids}

\title{Discriminative Measures for Comparison of Phylogenetic Trees}


\author{Omur Arslan}
\ead{omur@seas.upenn.edu}
% To refer to this use \corref{mycorrespondingauthor} 
%\cortext[mycorrespondingauthor]{Corresponding author.}

\author{Dan P. Guralnik}
\ead{guralnik@seas.upenn.edu}

\author{Daniel E. Koditschek}
\ead{kod@seas.upenn.edu}

\address{Department of Electrical and Systems Engineering, University of
Pennsylvania, Philadelphia, PA 19104, USA}

% To refer to this footnote, use \fnref{myfootnote}
%\fntext[myfootnote]{This work was funded in part by the Air Force Office of Science Research under the MURI FA9550-10-1-0567.}

\begin{abstract}
In this paper we introduce and study three new measures for efficient discriminative comparison of phylogenetic trees. 
The \emph{NNI navigation dissimilarity} $\distNav$ counts the steps along a ``combing'' of the Nearest Neighbor Interchange (NNI) graph of binary hierarchies, providing an efficient approximation to the (NP-hard) NNI distance in terms of ``edit length''. At the same time, a closed form formula for $\distNav$ presents it as a weighted count of pairwise incompatibilities between clusters, lending it the character of an edge dissimilarity measure as well. A relaxation of this formula to a simple count yields another measure on {\it all} trees~ ---~ the \emph{crossing dissimilarity} $\dist_{CM}$.
Both dissimilarities are symmetric and positive definite (vanish only between identical trees) on binary hierarchies but they fail to satisfy the triangle inequality. Nevertheless, both are bounded below by the widely used Robinson-Foulds metric and bounded above by a closely related true metric, the \emph{cluster-cardinality metric} $\dist_{CC}$.  
We show that each of the three proposed new dissimilarities is computable in time $\bigO{n^2}$ in the number of leaves $n$, and conclude the paper with a brief numerical exploration of the distribution over tree space of these dissimilarities in comparison with the Robinson-Foulds metric and the more recently introduced matching-split distance.
\end{abstract}


\end{frontmatter}

%\linenumbers

%%%%%%%%%%%%%%%%%%%%%%%%%%%%%%%%%%%%%%%%%%
%%%%%%%%%%%%%%%%%%%%%%%%%%%%%%%%%%%%%%%%%%
\section{Introduction}
\label{sec.Introduction}
%%%%%%%%%%%%%%%%%%%%%%%%%%%%%%%%%%%%%%%%%%
%%%%%%%%%%%%%%%%%%%%%%%%%%%%%%%%%%%%%%%%%%

%%%%%%%%%%%%%%%%%%%%%%%%%%%%%%%%%%%%%%%%%%
%%%%%%%%%%%%%%%%%%%%%%%%%%%%%%%%%%%%%%%%%%
\subsection{Motivation}
%%%%%%%%%%%%%%%%%%%%%%%%%%%%%%%%%%%%%%%%%%
%%%%%%%%%%%%%%%%%%%%%%%%%%%%%%%%%%%%%%%%%%

A fundamental classification problem common to both computational biology and engineering is the efficient and informative comparison of hierarchical structures.
In bioinformatics settings, these typically take the form of phylogenetic trees representing evolutionary relationships within a set $\indexset$ of taxa. 
In pattern recognition and data mining settings, hierarchical trees are often used to encode nested sequences of groupings of a set of observations. 
Dissimilarity between combinatorial trees has been measured in the past literature largely by recourse to one of two separate approaches: comparing edges and counting edit distances.
Representing the former approach, a widely used tree metric is the Robinson-Foulds (RF) distance, $\distRF$, \cite{robinson_foulds_mb1981} whose count of the disparate edges between trees requires linear time, $\bigO{n}$, in the number of leaves, $n$, to compute \cite{day_joc1985}.
Empirically,  $\distRF$ offers only a very coarse measure of disparity, and among its many proposed refinements, the recent matching split distance $\distMS$, \cite{bogdanowicz_giaro_TCBB2012, lin_rajan_moret_TCBB2012} offers a more discriminative metric albeit with considerably higher computational cost, $\bigO{n^{2.5} \log n}$.
Alternatively, various edit distances have been proposed \cite{robinson_jct1971, moore_goodman_barnabas_jtb1973, allen_steel_AC2001, felsenstein2004} but the most natural variant, the Nearest Neighbor Interchange (NNI) distance $\distNNI$, entails an NP-complete computation for both labelled and unlabelled trees \cite{dasguptaEtAl_SDM1997}. 

%%%%%%%%%%%%%%%%%%%%%%%%%%%%%%%%%%%%%%%%%%
%%%%%%%%%%%%%%%%%%%%%%%%%%%%%%%%%%%%%%%%%%
\subsection{Results}
%%%%%%%%%%%%%%%%%%%%%%%%%%%%%%%%%%%%%%%%%%
%%%%%%%%%%%%%%%%%%%%%%%%%%%%%%%%%%%%%%%%%%

Our main contribution is the introduction of a dissimilarity measure on the space $\BT{\indexset}$ of labelled binary trees which bridges the above approaches by what is, effectively, a solution to the {NNI navigation problem} in $\BT{\indexset}$:
\begin{problem}[NNI Navigation Problem]\label{navigation problem} Given a target $\treeB\in\BT{\indexset}$, provide an efficient algorithm $\mathcal{A}_\treeB$ which, for any $\treeA\in\BT{\indexset}$, computes a Nearest Neighbor Interchange to be performed on $\treeA$ while guaranteeing that successive application of $\mathcal{A}_\treeB$  terminates in $\treeB$.
\end{problem}
\noindent This problem is motivated by applications in coordinated robot navigation \cite{arslanEtAl_Allerton2012, arslan_guralnik_kod_WAFR2014, ayanian_kumar_koditschek_RR2011}, where a group of robots is required to reconfigure reactively in real time their (structural) adjacencies while navigating towards a desired goal configuration. Thus, our particular formulation of the problem is inspired by the notion of reactive planning \cite{burridge_rizzi_kod_IJRR1999}, but may likely hold value for researchers interested in tree consensus and averaging  as well.

Of course, since computation of $\distNNI$ is NP-hard, one cannot hope for repeated applications of $\mathcal{A}_\treeB$ to produce NNI geodesics without incurring prohibitive complexity in each iteration. However, as we will show, constructing an efficient navigation scheme is possible if we allow the algorithm to produce less restricted paths: for $\card{\indexset}=n$, our navigation algorithms require $\bigO{n}$ time for each iteration and produce paths of length $\bigO{n^2}$ (as compared to the $\bigO{n\log n}$ diameter of $\distNNI$~ ---~ see \eqref{eq:diameters}).

Additional insight into the geometry of the space $(\BT{\indexset},\distNNI)$ is gained by recognizing a significant degree of freedom with which our navigation algorithm may select the required tree restructuring operation at each stage. As it turns out, for any given target $\treeB$, the repeated application of $\mathcal{A}_\treeB$ to a tree $\treeA$ until reaching $\treeB$ will yield paths of equal lengths regardless of any choices made along the way. This length, by definition, is the navigation dissimilarity $\distNav\prl{\treeA,\treeB}$ (and is obtained, in the manner described, in $\bigO{n^3}$ time, though more efficient implementations will guarantee $\bigO{n^2}$). At the same time, a closed form formula we derive for $\distNav$ allows us to avoid computing a navigation path when only the value of $\distNav$ is needed, and computes it in $\bigO{n^2}$ time. Surprisingly, despite the asymmetric character of its construction, $\distNav$ is a symmetric (and positive definite) dissimilarity on $\BT{\indexset}$, though it fails to be a metric.

Although $\distNav$ does not satisfy the triangle inequality, it is related to the well accepted Robinson-Foulds distance by the following tight bounds:
\begin{equation}\label{eq:dnav vs dRF}
 \distRF \leq \distNav \leq  \frac{1}{2}\distRF^2 + \frac{1}{2} \distRF ~,
\end{equation}
We find it useful to introduce a ``relaxation'' of $\distNav$, the \emph{crossing dissimilarity} $\distCM$. This dissimilarity simply counts all the pairwise cluster incompatibilities between two trees, hence it is symmetric, positive-definite, and computable in $\bigO{n^2}$ time. In fact, the two dissimilarities are commensurable, leading to similar bounds in terms of $\distRF$:
\begin{equation}
\distCM \leq \distNav \leq \frac{3}{2} \distCM\,,\quad
\distRF \leq \distCM \leq \distRF^2 ~.
\end{equation}
Finally, we introduce  a true metric whose spatial resolution and computational complexity is comparable to those our new dissimilarities. Exploiting a well known relation between trees and ultrametrics \cite{carlsson_memoli_jmlr2010}, we also introduce \emph{the cluster-cardinality distance} $\distCC$~ ---~ constructed as the pullback of a matrix norm along an embedding of hierarchies into the space of matrices and computable in $\bigO{n^2}$ time~ ---~ which is a true metric bounding $\distCM$ from above (and hence also $\distNav$, up to a constant factor). Thus, cumulatively we obtain:
\begin{equation}\label{eq.TreeMeasureOrder}
	\frac{2}{3}\distRF\leq 
	\frac{2}{3}\distNav\leq 
	\distCM\leq
	\distCC ~.
\end{equation}
%
%To summarize, in a manner of speaking, the dissimilarities $\distNav$ and $\distCM$ are not that far from being metrics as one might have worried.

\medskip
We have surveyed some of the new features of our tree proximity measures that might hold interest for pattern classification and phylogeny analysis relative to the diverse alternatives that have appeared in the literature.  
Closest among these many alternatives \cite{li_tromp_zhang_jtb1996, culik_wood_IPL1982, brown_day_JOC1984}, $\distNav$ has some resemblance to an early NNI graph navigation  algorithm, $\dist_{ra}$ \cite{brown_day_JOC1984} which used a divide-and-conquer approach with a balancing strategy to achieve an $\bigO{n \log n}$ computation of tree dissimilarity.
Notwithstanding its lower computational cost, in contrast to $\distNav$, the recursive definition of $\dist_{ra}$, as with many NNI distance approximations \cite{li_tromp_zhang_jtb1996, culik_wood_IPL1982, brown_day_JOC1984}, does not admit a closed form expression. % (and, likely in consequence, enjoys no reported metric upper bound).

It is often of interest to compare more than pairs of hierarchies at a time, and the notion of a ``consensus'' tree has accordingly claimed a good deal of attention in the literature  \cite{bryant_DIMACS2003}. 
For instance, the majority rule tree \cite{margush_mcmorris_BMB1981} of a set of  trees is a median tree respecting the RF distance and provides statistics on the central tendency of trees \cite{barthelemy_mcmorris_JC1986}.
When $\distNav$ and $\distCM$ are extended to degenerate trees they fail to be positive definite, and thus their behavior over (typically degenerate) consensus trees departs still further from the properties of a true metric. 
However, it turns out that both notions of a consensus tree (strict \cite{rohlf_MB1982}, and loose/semi-strict \cite{bremer_C1990}) behave as median trees with respect to both our dissimilarities. 
In fact, the loose consensus tree is the maximal (finest) median tree with respect to inclusion for both $\distNav$ and $\distCM$.

%A final observation of significant interest in some application settings   is that the computation  of $\distNav$  derives from an exact path in tree space that can be explicitly computed with the same $\bigO{n^2}$ computational cost.

The paper is organized as follows.
\refsec{sec.prelim} briefly summarizes the necessary background while introducing the notation used throughout the sequel.
\refsec{sec.incompatibility} introduces and studies the cluster-cardinality distance $\distCC$ and the crossing dissimilarity $\distCM$.
In \refsec{sec.dnav} we present a solution of the NNI navigation problem and study properties of the resulting NNI navigation dissimilarity $\distNav$ and its relations with other tree dissimilarity measures.
\refsec{sec:Discussion} discusses the relation between commonly used consensus models and our tree dissimilarities $\distCM$ and $\distNav$, and compares our proposed tree measures with $\distRF$ and $\distMS$ based on some frequently used empirical distributions of tree measures.
A brief discussion of future directions follows in \refsec{sec.Conclusion}.

%%%%%%%%%%%%%%%%%%%%%%%%%%%%%%%%%%%%%
%%%%%%%%%%%%%%%%%%%%%%%%%%%%%%%%%%%%%
\section{Preliminaries}
\label{sec.prelim}
%%%%%%%%%%%%%%%%%%%%%%%%%%%%%%%%%%%%%
%%%%%%%%%%%%%%%%%%%%%%%%%%%%%%%%%%%%%

%We now introduce our basic notation used throughout the paper and recall several standard notions of hierarchies, such as cluster compatibility, hierarchical relations of clusters and tree operations, from a set theoretical perspective.

%%%%%%%%%%%%%%%%%%%%%%%%%%%%%%%%%%%%%%%
%%%%%%%%%%%%%%%%%%%%%%%%%%%%%%%%%%%%%%%
\subsection{Hierarchies}
\label{sec.Hierarchies}
%%%%%%%%%%%%%%%%%%%%%%%%%%%%%%%%%%%%%%%
%%%%%%%%%%%%%%%%%%%%%%%%%%%%%%%%%%%%%%%

By a \emph{hierarchy} $\tree$ over a fixed non-empty finite index set $\indexset$ we shall mean a rooted tree with labeled leaves (see \reffig{fig:hierarchicalrelation}). Formally, $\tree$ is a finite connected acyclic graph with leaves (vertices of degree one) bijectively labelled by $\indexset$, and edges oriented in such a way that (i) all interior vertices have out-degree at least two, and (ii) there is a vertex, referred to as the \emph{root of $\tree$}, such that every edge is oriented away from the root. Under these assumptions all the vertices of $\tree$ are reachable from the root through a directed path in $\tree$ \cite{billera_holmes_vogtmann_aap2001}.

The \emph{cluster} $\cluster{v}$ of a vertex $v \in V_{\tree}$ of a hierarchy $\tree$ is defined to be the set of leaves reachable from $v$ by a directed path in $\tree$.
Singleton clusters and the root cluster $\indexset$ are common to all trees, and we refer to them as the trivial clusters. 
We denote by $\cluster{\tree}$ (respectively $\intcluster{\tree}$) the set of all clusters (resp. non-trivial clusters) of $\tree$:
\begin{equation}
\cluster{\tree} \ldf \crl{\cluster{v} \big | \, v \in V_{\tree}} \subseteq \PowerSet{\indexset} ~, \qquad
\intcluster{\tree} \ldf \crl{I \in \cluster{\tree} \setminus \crl{\indexset} \Big | \card{I} \geq 2} ~,
\end{equation}    
where $\PowerSet{\indexset}$ denotes the power set of $\indexset$.

%%%%%%%%%%%%%%%%%%%%%%%%%%%%%%%%%%%%%%%%%
%%%%%%%%%%%%%%%%%%%%%%%%%%%%%%%%%%%%%%%%%
\subsubsection{Compatibility}
\label{sec.compatibility} 
%%%%%%%%%%%%%%%%%%%%%%%%%%%%%%%%%%%%%%%%%
%%%%%%%%%%%%%%%%%%%%%%%%%%%%%%%%%%%%%%%%%

\begin{definition}[\cite{schrijver2003, felsenstein2004}]
\label{def.compatibility}
Subsets $A,B\subset\indexset$ are said to be \emph{compatible}, $A \compatible B$, if
\begin{equation}
A \cap B = \varnothing \; \logicor  \; A \subseteq B  \; \logicor  \; B \subseteq A ~.
\end{equation}
If $A\not\compatible B$, then we say that \emph{$A$ and $B$ cross}. We further extend the compatibility relation $(\compatible)$ as follows:
\begin{itemize}
	\item For $\mathcal{A},\mathcal{B}\subseteq\PowerSet{\indexset}$, write $\mathcal{A}\compatible\mathcal{B}~$ if $~A \compatible B$  for all $A\in \mathcal{A}$ and $B \in \mathcal{B}$;
	\item For a cluster $I\subseteq\indexset$ and a tree $\tree$ over the leaf set $\indexset$, write $I\compatible\tree~$ if $~\crl{I}\compatible\cluster{\tree}$;
	\item For two trees $\treeA$ and $\treeB$ over the leaf set $\indexset$, write $\treeA\compatible\treeB~$ if $~\cluster{\treeA} \compatible \cluster{\treeB}$.
\end{itemize}
\end{definition}

By construction, any two elements of $\cluster{\tree}$ are compatible  for any tree $\tree$. This motivates the following definition: 
\begin{definition}[\cite{schrijver2003}]\label{def.nested}
A subset $\mathcal{A}$ of $\PowerSet{\indexset}$ is said to be \emph{nested}~ ---~ also referred to in the literature as a ``laminar family''~ ---~ if any two elements of $\mathcal{A}$ are compatible. $\cluster{\tree}$ is known as the laminar family associated with $\tree$ .
\end{definition}

%%%%%%%%%%%%%%%%%%%%%%%%%%%%%%%%%%%%%%%%%%%%%
%%%%%%%%%%%%%%%%%%%%%%%%%%%%%%%%%%%%%%%%%%%%%
\subsubsection{Hierarchical Relations}
\label{sec.HierarchicalRelations }
%%%%%%%%%%%%%%%%%%%%%%%%%%%%%%%%%%%%%%%%%%%%%
%%%%%%%%%%%%%%%%%%%%%%%%%%%%%%%%%%%%%%%%%%%%%

\begin{figure}
\centering
\includegraphics[width=0.45\textwidth]{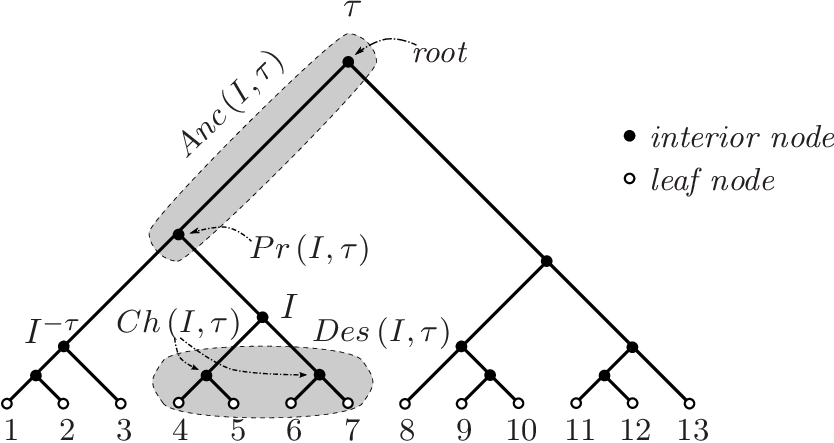} 
%\vspace{-2mm}
\caption{Hierarchical Relations: ancestors - $\ancestorCL{I,\tree}$, parent - $\parentCL{I,\tree}$, children - $\childCL{I,\tree}$,  descendants - $\descendantCL{I,\tree}$, and local complement (sibling) - $\complementLCL{I}{\tree}$  of cluster $I$ of a rooted binary phylogenetic tree, $\tree \in \bintreetopspace_{\brl{13}}$. 
Filled and unfilled circles represent interior and leaf nodes, respectively. 
An interior node is referred to by its cluster, the list of leaves below it; for example, $I = \crl{4,5,6,7}$. 
Accordingly, the cluster set of $\tree$ is $\cluster{\tree} = \left \{ \big. \right.\!\! \crl{1}, \crl{2}, \ldots, \crl{13}, \crl{1,2}, \crl{1,2,3}, \crl{4,5}, \crl{6,7}, \crl{4,5,6,7}, \crl{1,2, \ldots, 7}, \crl{9, 10},  \crl{8,9,10},  \crl{11,12}, \crl{11,12,13}, \crl{8,9, \ldots, 13}, \crl{1,2, \ldots, 13} \!\! \left . \big.\right\}$.
}
\label{fig:hierarchicalrelation}
\vspace{-1mm}
\end{figure}

The cluster set $\cluster{\tree}$ of a hierarchy $\tree$ completely determines its representation as a rooted tree with labeled leaves: $\cluster{\tree}$ stands in bijective correspondence with the vertex set of $\tree$, and $(v, v')$ is an edge in $\tree$ if and only if $\cluster{v}\supset \cluster{v'}$ and there is no $\tilde{v} \in V_{\tree}$ such that $\cluster{v}\supset \cluster{\tilde{v}} \supset \cluster{v'}$. Consequently, the standard notions of ancestor, descendant, parent and child of a vertex in common use for rooted trees carry over to the cluster representation as follows:

\begin{subequations}
\noindent
\begin{align}
\ancestorCL{I,\tree} &= \left \{ V \hspace{-0.5mm} \in \hspace{-0.5mm} \cluster{\tree} \big | \, I \subsetneq V \right \},
&
\descendantCL{I,\tree} &= \crl{V \hspace{-0.5mm} \in \hspace{-0.5mm} \cluster{\tree}  \big | \, V \subsetneq I},\\
\parentCL{I,\tree}  
 &= \min\prl{\ancestorCL{I,\tree}},% \setminus \bigcup_{A \in \ancestorCL{I,\tree}} \ancestorCL{A,\tree},
&
\childCL{I, \tree} &= \crl{V \in \cluster{\tree} \big |\, \parentCL{V,\tree} = I },
\end{align}
\end{subequations}
where $\min\prl{\ancestorCL{I,\tree}}$ is computed with respect to the inclusion order. Note that for the trivial clusters we have $\parentCL{\indexset, \tree} = \varnothing$ and $\childCL{\{s\},\tree}=\varnothing$ for $s\in\indexset$.

Since the set of children partitions each parent, we find it useful to define the \emph{local complement} $\complementLCL{I}{\tree}$ of $I \in \cluster{\tree}$ as
\begin{equation}\label{eq:compLCL}
\complementLCL{I}{\tree} \ldf \parentCL{I,\tree}\setminus I ~,
\end{equation}
not to be confused with the standard (global) complement, $\complementCL{I}= \indexset \setminus I$.
Further, a grandchild in $\tree$ is a cluster $G \in \cluster{\tree}$ having a grandparent $\grandparentCL{G,\tree} \ldf \parentCL{\big. \parentCL{G,\tree}, \tree}$ in $\tree$. 
We denote the set of all grandchildren in $\tree$  by $\grandchildset{\tree}$,
\begin{equation} \label{eq.grandchild}
\grandchildset{\tree} \ldf \crl{G \in \cluster{\tree} \big | \, \grandparentCL{G,\tree} \neq \varnothing} ~.
\end{equation} 
If $A,B$ are either elements of $\indexset$ or clusters of $\tree$, it is convenient to have $\cancCL{A}{B}{\tree}$ denote the smallest (in terms of cardinality) common ancestor of $A$ and $B$ in $\tree$. Finally, the depth $\depth{\tree}(I)$ of a cluster in a hierarchy $\tree$ is defined to equal the number of distinct ancestors of $I$ in $\tree$.

%%%%%%%%%%%%%%%%%%%%%%%%%%%%%%%%%%%%%%%%%%
%%%%%%%%%%%%%%%%%%%%%%%%%%%%%%%%%%%%%%%%%%
\subsubsection{Nondegeneracy}
\label{sec.nondegeneracy}
%%%%%%%%%%%%%%%%%%%%%%%%%%%%%%%%%%%%%%%%%%
%%%%%%%%%%%%%%%%%%%%%%%%%%%%%%%%%%%%%%%%%%

A rooted tree where every interior vertex has exactly two children is said to be \emph{binary} or \emph{non-degenerate}. All other trees are said to be \emph{degenerate}. 
We will denote the set of hierarchies over a finite leaf set $\indexset$, by $\treetopspace_{\indexset}$. The subset of non-degenerate hierarchies will be denoted by $\bintreetopspace_{\indexset}$.

Note that the laminar family $\cluster{\tree}$ of a degenerate tree $\tree$ may always be augmented with additional clusters while remaining nested (\refdef{def.nested}). This leads to the well known result:
\begin{remark}[\cite{vogtmann_2007, schrijver2003}]\label{rem:MaximumCardinality}
Let $\tree \in \treetopspace_{\indexset}$. 
Then $\tree$ has at most $2\card{\indexset}-1$ vertices, with equality if and only if $\tree$ is nondegenerate, if and only if $\cluster{\tree}$ is a maximal laminar family in $\PowerSet{\indexset}$ with respect to inclusion.\footnote{In this paper we adopt the convention that a laminar family does not contain the empty set (as an element).}
\end{remark}

%
%\begin{definition}[\cite{owen_provan_tcbb2011}]\label{def:DisjointTree}
%Hierarchies $\treeA, \treeB \in \treetopspace_{\indexset}$ are said to be \emph{disjoint} if they have no non-trivial clusters in common.
%\end{definition}
%

%%%%%%%%%%%%%%%%%%%%%%%%%%%%%%%%%%%%%%%%%%
%%%%%%%%%%%%%%%%%%%%%%%%%%%%%%%%%%%%%%%%%%
\subsubsection{Consensus}
\label{sec.consensus}
%%%%%%%%%%%%%%%%%%%%%%%%%%%%%%%%%%%%%%%%%%
%%%%%%%%%%%%%%%%%%%%%%%%%%%%%%%%%%%%%%%%%%

\begin{definition}[\cite{rohlf_MB1982, bremer_C1990}] \label{def.StrictLooseConsensus}
For any set of trees $T$ in $\treetopspace_{\indexset}$, the strict and loose consensus trees of $T$, denoted $T_{\ast}$ and $T^\ast$ respectively, are defined by specifying their cluster sets as follows:
\begin{equation}
\cluster{T_{\ast}} = \bigcap_{\tree \in T} \cluster{\tree}\,,\quad
\cluster{T^{\ast}} = \crl{I \in \bigcup_{\treeB \in T} \cluster{\treeB} \Bigg | \, \forall \treeA \in T  \quad I \compatible \treeA  }.
\end{equation}
\end{definition}

\noindent Note that the loose consensus tree $T^{\ast}$ of $T$ refines the strict consensus tree $T_{\ast}$, that is $\cluster{T^{\ast}} \supseteq \cluster{T_{\ast}}$.

%%%%%%%%%%%%%%%%%%%%%%%%%%%%%%%%%%%%%%%%%%%%
%%%%%%%%%%%%%%%%%%%%%%%%%%%%%%%%%%%%%%%%%%%%
\subsection{Some Operations on Trees}
\label{sec.TreeOperations}
%%%%%%%%%%%%%%%%%%%%%%%%%%%%%%%%%%%%%%%%%%%
%%%%%%%%%%%%%%%%%%%%%%%%%%%%%%%%%%%%%%%%%%%

%%%%%%%%%%%%%%%%%%%%%%%%%%%%%%%%%%%%%%%%%%%
%%%%%%%%%%%%%%%%%%%%%%%%%%%%%%%%%%%%%%%%%%%
\subsubsection{The NNI Graph}
\label{sec.NNIGraph} 
%%%%%%%%%%%%%%%%%%%%%%%%%%%%%%%%%%%%%%%%%%%
%%%%%%%%%%%%%%%%%%%%%%%%%%%%%%%%%%%%%%%%%%%

The standard definition of NNI walks on unrooted binary trees \cite{robinson_jct1971, moore_goodman_barnabas_jtb1973} conveniently restricts to the space $\bintreetopspace_{\indexset}$ of rooted binary trees as follows:

\begin{definition}\label{def:NNIMove}
Let $\treeA \in \bintreetopspace_{\indexset}$.
We say that $\treeB \in \bintreetopspace_{\indexset}$ is the result of performing a \emph{Nearest Neighbor Interchange (NNI) move} on $\treeA$ at a grandchild $G \in \grandchildset{\treeA}$ \refeqn{eq.grandchild} if  
\begin{equation}
\cluster{\treeB} = \prl{ \Big. \cluster{\treeA} \setminus \crl{\big.\parentCL{G,\treeA}}} \cup \crl{\big.\grandparentCL{G,\treeA} \setminus G}. \label{eq.NNImove}
\end{equation}
We often indicate this by writing $\treeB=\NNI(\treeA,G)$.
\end{definition}
%\begin{remark} We will not be considering other tree-restructuring operations in this work, and will henceforth refer to NNI moves simply as "moves".
%\end{remark}
\noindent Note that the NNI move at cluster $G $ on $\treeA$ swaps cluster $G$ with its parent's sibling $\complementLCL{\parentCL{G,\treeA}}{\treeA}$ to yield $\treeB$, depicted in \reffig{fig.NNImoveNNIGraph}(left); and after an NNI move at cluster $G$ of $\treeA$, grandchild $G$ of grandparent $P =  \grandparentCL{G,\treeA}$   with respect to $\treeA$ becomes  child $G$ of parent $P = \parentCL{G,\treeB}$ with respect to $\treeB$.

It is standard to say that $\treeA,\treeB \in \BT{\indexset}$ are NNI-adjacent if and only if one can be obtained from the other by a single move. \reffig{fig.NNImoveNNIGraph}(left) illustrates the moves on $\BT{\indexset}$ and their inverses. 

\begin{figure}[htb]
\centering
\begin{tabular}{c@{\hspace{15mm}}c}
\raisebox{2mm}{\includegraphics[width=0.35\textwidth]{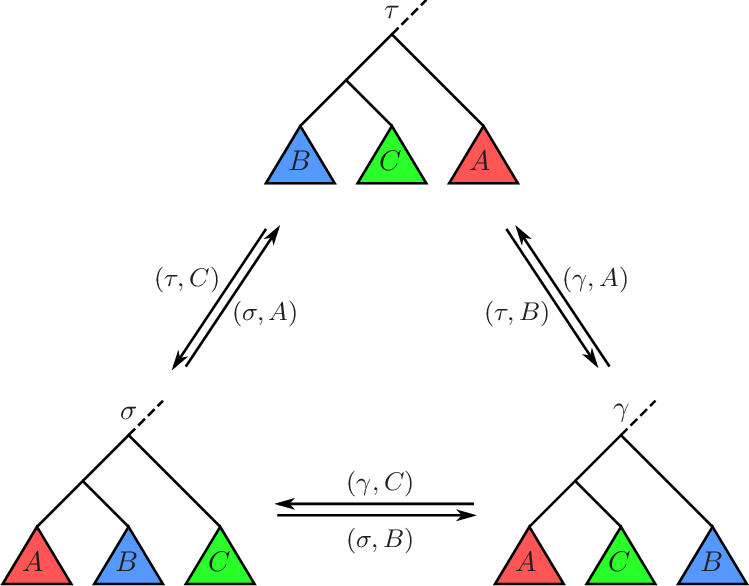}}  & 
\includegraphics[width=0.445\textwidth]{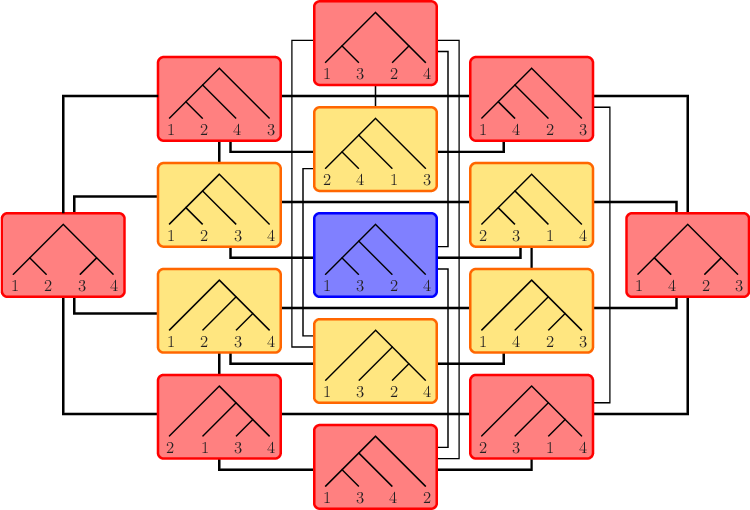} 
\end{tabular}
%\vspace{-2mm}
\caption{NNI moves (arrows, left) between binary trees, each move is labeled by its source tree and the grandchild defining the move, and
the NNI Graph for $\indexset = [4] = \crl{1,2,3,4}$ (right).}
\label{fig.NNImoveNNIGraph}
%\vspace{-1mm}
\end{figure}

The NNI-graph is formed over the vertex set $\BT{\indexset}$ by declaring two trees to be connected by an edge if and only if they are NNI-adjacent, see e.g. \reffig{fig.NNImoveNNIGraph}(right). We will work with a directed version of this graph:
\begin{definition}\label{def:NNIGraph}
The \emph{directed NNI graph} $\NNIgraph_{\indexset} = \prl{\BT{\indexset}, \NNIedgeset_{\indexset}}$ is the directed graph on $\BT{\indexset}$ with $(\treeA,\treeB)\in\NNIedgeset_{\indexset}$ iff $\treeB$ results from applying an NNI move to $\treeA$. We will henceforth identify the notation for an NNI move $(\treeA,G)$, $G\in\grandchildset{\treeA}$ with the directed edge $\prl{\treeA,NNI(\treeA,G)}\in\NNIedgeset_{\indexset}$ wherever there is no danger of confusion.
\end{definition}

The (directed) NNI-graph on $n$ leaves is a regular graph of out-degree $2(n-2)$ \cite{robinson_jct1971}. Our description clarifies this by parametrizing the set of neighbors of $\tree\in\bintreetopspace_{\indexset}$ with its grandchildren, $\card{\grandchildset{\tree}} = 2(\card{\indexset} - 2)$. The vertex set of the NNI graph is known to grow super exponentially with the number of leaves~\cite{billera_holmes_vogtmann_aap2001},
\begin{equation}\label{eqn:NumEvolutionaryTree}
	\card{\bintreetopspace_{[n]}} = (2n - 3)!! = (2n-3)(2n-5) \ldots 3 ~,\quad n \geq 2 ~. 
\end{equation}  
As a result, exploration of the NNI-graph (for example, searching for the shortest path between hierarchies or an optimal phylogenetic tree model) rapidly becomes impractical and costly as the number of leaves increases.
A useful observation for NNI-adjacent trees is: 
\begin{lemma}\label{lem.NNITriple}
An ordered pair of hierarchies $\prl{\treeA,\treeB}$ is an edge in $\NNIgraph_{\indexset}$ if and only if there exists an ordered triple $\prl{A,B,C}$ of common clusters of $\treeA$ and $\treeB$ such that  $\crl{A \cup B} = \cluster{\treeA} \setminus \cluster{\treeB}$ and $\crl{B\cup C} = \cluster{\treeB}\setminus \cluster{\treeA}$. 
The triple $\prl{A,B,C}$ is uniquely determined by $\prl{\treeA,\treeB}$ and will be referred to as  the \emph{NNI-triplet} associated with  $\prl{\treeA,\treeB}$.
\end{lemma}
\begin{proof} The proof amounts to a formal restatement of the observations made in \reffig{fig.NNImoveNNIGraph}(left). See \refapp{app.NNITriple}.
\end{proof}

\noindent Observe that the triplet in reverse order $\prl{C,B,A}$ is the NNI-triple associated with the edge $\prl{\treeB, \treeA}$. 
Also note that the NNI moves on $\treeA$ at $A$ and on $\treeB$ at $C$ yield $\treeB$ and $\treeA$, respectively.

%%%%%%%%%%%%%%%%%%%%%%%%%%%%%%%%%%%%%%%%%%%
%%%%%%%%%%%%%%%%%%%%%%%%%%%%%%%%%%%%%%%%%%%
\subsubsection{Tree Restriction}
\label{sec.TreeRestriction} 
%%%%%%%%%%%%%%%%%%%%%%%%%%%%%%%%%%%%%%%%%%%
%%%%%%%%%%%%%%%%%%%%%%%%%%%%%%%%%%%%%%%%%%%

\begin{definition}\label{def:TreeRestriction}
Let $\indexset$ be a fixed finite set and $K \subseteq \indexset$.
The restriction map $\treeres{K}:\PowerSet{\indexset} \rightarrow \PowerSet{K}$ is defined to be
\begin{equation}\label{eq.TreeRestriction}
\treeres{K}\prl{\mathcal{A}} \ldf \crl{A \cap K  \, \big | \, A \in \mathcal{A} \, , \; A \cap K \neq \varnothing}
\end{equation}
for any $\mathcal{A} \subseteq \PowerSet{\indexset}$.
It is convenient to have $\treeressym{K}{\mathcal{A}}$ denote $\treeres{K}\prl{\mathcal{A}}$.
For $\treeA \in \treetopspace_{K}$ and $\treeB \in \treetopspace_{\indexset}$ we will write:
\begin{equation}
\treeA= \treeres{K}\prl{\treeB} = \treeressym{K}{\treeB}  \Longleftrightarrow \cluster{\treeA}= \treeressym{K}{\cluster{\treeB}} ~.
\end{equation}
%
%and we will have $\treeressym{K}{\treeB}$ denote $\treeres{K}\prl{\treeB}$.
\end{definition}

\begin{remark} \label{rem:ClusterConcatenation}
Let $\tree \in \bintreetopspace_{\indexset}$ and $\crl{\indexset_L, \indexset_R} = \childCL{\indexset, \tree}$. 
Then one has  $\cluster{\tree} = \cluster{\treeressym{\indexset_L}{\tree}} \cup \crl{\big.\indexset} \cup \cluster{\treeressym{\indexset_{R}}{\tree}}$ .
%
%\begin{equation} \label{eq:ClusterConcatenation}
%\cluster{\tree} = \cluster{\treeressym{\indexset_L}{\tree}} \cup \crl{\big.\indexset} \cup \cluster{\treeressym{\indexset_{R}}{\tree}}. 
%\end{equation}  
\end{remark}

\begin{lemma}\label{lem:BinTreeRestriction}
For any  finite set $\indexset$ and  $K \subseteq \indexset$ with $\card{K}\geq 2$, $\treeres{K}\prl{\bintreetopspace_{\indexset}}= \bintreetopspace_{K}$. 
\end{lemma}
\begin{proof} See \refapp{app.BinTreeRestriction}.
\end{proof}

%%%%%%%%%%%%%%%%%%%%%%%%%%%%%%%%%%%%%%%%%%%%%%%%%%%%%%%
%%%%%%%%%%%%%%%%%%%%%%%%%%%%%%%%%%%%%%%%%%%%%%%%%%%%%%%
\subsection{Dissimilarities, Metrics and Ultrametrics}
\label{sec.Dissimilarity}
%%%%%%%%%%%%%%%%%%%%%%%%%%%%%%%%%%%%%%%%%%%%%%%%%%%%%%%
%%%%%%%%%%%%%%%%%%%%%%%%%%%%%%%%%%%%%%%%%%%%%%%%%%%%%%%

Recall that a \emph{dissimilarity measure} on $X$, or simply a \emph{dissimilarity}, is a real-valued nonnegative symmetric function $\dist$ on $X\times X$ satisfying $\dist\!\prl{x,x} = 0$ for all $x\in X$. %\footnote{\ghost{We use $\R$ and $\R_{\geq 0}$ to denote the real line and the set of nonnegative reals, respectively.}}
Recall that a dissimilarity $\dist$ on $X$ is \emph{positive definite} if $\dist\prl{x,y} = 0$ implies $x=y$ for all $x,y \in X$. Many approximations of the (NP-hard) NNI metric are positive definite dissimilarities \cite{li_tromp_zhang_jtb1996, culik_wood_IPL1982, brown_day_JOC1984}. A dissimilarity $\dist$ is a metric if it satisfies the triangle inequality, $\dist\prl{x,y} \leq \dist\prl{x,z} + \dist\prl{z,y}$ for all $x,y,z \in X$. For example:
\begin{definition}[\cite{robinson_foulds_mb1981} and\cite{bogdanowicz_giaro_TCBB2012,lin_rajan_moret_TCBB2012}] \label{def.RF and MS distances}
The Robinson-Foulds distance $\distRF$ on $\treetopspace_{\indexset}$ is defined by: \footnote{Here, $\ominus$ denotes the symmetric set difference, i.e. $A \ominus B = \prl{A \setminus B} \cup \prl{B \setminus A}$ for any sets $A$ and $B$. }
\begin{equation}\label{eq.RFdistance}
\distRF\prl{\treeA,\treeB} = \frac{1}{2} \card{\big. \cluster{\treeA} \ominus \cluster{\treeB}}\,,\quad\treeA,\treeB \in \treetopspace_{\indexset} ~. 
\end{equation}  

The \emph{matching split distance} $\distMS$ between a pair of hierarchies $\treeA$ and $\treeB$ in $\BT{\indexset}$ is defined to be the value of a minimum-weighted perfect matching in the graph $G_\indexset\!\prl{\treeA,\treeB}$ obtained from $\treeA, \treeB \in \BT{\indexset}$ as the complete bipartite graph with sides $\intcluster{\treeA}$ and $\intcluster{\treeB}$ with each edge $\prl{I,J} \in \intcluster{\treeA} \times \intcluster{\treeB}$ carrying the weight \footnote{This corresponds to the Hamming distance of clusters.} $A_S\prl{I,J} = \min\prl{\card{\Big.I \ominus J} , \card{I\ominus \complementCL{J}} }$.
\end{definition}

It is known that  $\distRF \leq \distMS \leq \frac{\card{\indexset} + 1}{2} \distRF$ \cite{bogdanowicz_giaro_TCBB2012}, which explains the improvement of $\distMS$ over $\distRF$ in discriminative power. At the same time, the cost of computing a minimum weighted perfect matching in any $G_\indexset\!\prl{\treeA,\treeB}$ is $\bigO{\card{\indexset}^{2.5}\log \card{\indexset}}$, which motivates the search for dissimilarities producing similar improvement in discriminative power (bounding $d_{RF}$ from above) yet having a lower computational cost than that of $\distMS$.

\medskip
Recall that an \emph{ultrametric} $\dist$ on $X$ is a metric on $X$ satisfying the strengthened triangle inequality, $\dist\prl{x,y} \leq \max\prl{\big. \dist\prl{x,z},\dist\prl{z,y}}$ for all $x,y,z \in X$. The following is a restatement of a well known fact (see, e.g. \cite{carlsson_memoli_jmlr2010, jain_dubes_1988, rammal_et_al_RMP1986}) revealing the relation between hierarchies and ultrametrics:
\begin{lemma}\label{lem.treeUltrametric}
Let $\tree\in\treetopspace_{\indexset}$ and $h_{\tree}:\cluster{\tree} \rightarrow \R_{\geq 0}$.
For any $i,j\in\indexset$ let $\cancCL{i}{j}{\tree}$ denote the smallest cluster in $\cluster{\tree}$ containing the pair $\{i,j\}$. 
Then the dissimilarity on $\indexset$ given by 
\begin{equation} \label{eq.treeUltrametric}
 \dist_{\tree}\prl{i,j}  \ldf h_{\tree} \prl{\big. \cancCL{i}{j}{\tree}}, \quad i,  j \in \indexset, 
\end{equation}  
is an ultrametric if and only if the following are satisfied for any $ I,J \in \cluster{\tree}$:
\begin{enumerate}[label=(\alph*)] 
\item \label{it.heightMonotonicity} if $I \subseteq J$, then $h_{\tree}\prl{I} \leq h_{\tree}\prl{J}$ ,
\item \label{it.heightIndiscernible}$h_{\tree}\prl{I} = 0$ if and only if $\card{I}= 1$ .
\end{enumerate}
\end{lemma}
\begin{proof} 
See \refapp{app.treeUltrametric}.
\end{proof}

Recall that a set $X$ may always inherit a metric from a metric space $\prl{Y,\dist_{Y}}$ by pullback: any injective map $f$ of $X$ into $Y$ yields a metric $\dist_{X}$ on $X$ defined by $\dist_{X}\prl{x_1, x_2} \ldf \dist_{Y}\prl{f\prl{x_1}, f\prl{x_2}}$ and known as the pullback $\dist_{X}=f^\ast\dist_{Y}$ of $\dist_Y$ along $f$. For example, the RF metric is a pullback: it is common knowledge that the set $F\prl{X}$ of all finite subsets of a set $X$ forms a metric space under the metric $\dist\prl{A,B}= \card{A \ominus B}$, which is one of the ways of defining Hamming distance; thus, the RF distance is (one half times) the pullback of this metric on $F\prl{\PowerSet{\indexset}}$ under the map $\tree \mapsto \cluster{\tree}$.

%%%%%%%%%%%%%%%%%%%%%%%%%%%%%%%%%%%%%%%%%%%%%%%%
%%%%%%%%%%%%%%%%%%%%%%%%%%%%%%%%%%%%%%%%%%%%%%%%
\section{Quantifying Incompatibility}%{Discriminative Comparison of Edges}
\label{sec.incompatibility}
%%%%%%%%%%%%%%%%%%%%%%%%%%%%%%%%%%%%%%%%%%%%%%%%
%%%%%%%%%%%%%%%%%%%%%%%%%%%%%%%%%%%%%%%%%%%%%%%%

%In this section, we shall introduce a new tree metric based on ultrametric representation of hierarchies and a dissimilarity measure counting pairwise cluster compatibilities of trees.

%%%%%%%%%%%%%%%%%%%%%%%%%%%%%%%%%%%%%%%%%
%%%%%%%%%%%%%%%%%%%%%%%%%%%%%%%%%%%%%%%%%
\subsection{The Cluster-Cardinality Distance}
\label{sec.CCdistance}
%%%%%%%%%%%%%%%%%%%%%%%%%%%%%%%%%%%%%%%%%
%%%%%%%%%%%%%%%%%%%%%%%%%%%%%%%%%%%%%%%%%
 
We now introduce an embedding of hierarchies into the space of matrices based on the relation between  hierarchies and ultrametrics, summarized in \reflem{lem.treeUltrametric}: 
\begin{definition}\label{def.UltrametricRepn}
The \emph{ultrametric representation} is the map $\mat{U} : \treetopspace_{\indexset} \rightarrow \R^{\card{\indexset}  \times \card{\indexset} }$ defined by $\mat{U}(\tree)_{ij} \ldf h \prl{ \big. \cancCL{i}{j}{\tree}}$, where $h : \PowerSet{\indexset} \rightarrow \N$ is set to be $h\prl{I} \ldf  \card{I} - 1$, $I \subseteq \indexset$.
\end{definition}
%
%\begin{definition}
%The \emph{ultrametric representation} is the map $\mat{U} : \treetopspace_{\indexset} \rightarrow \R^{\card{\indexset}  \times \card{\indexset} }$, defined by 
%%
%\begin{equation}
%\mat{U}(\tree)_{ij} \ldf h \prl{ \big. \cancCL{i}{j}{\tree}}, \label{eq.UltrametricMat}
%\end{equation}
%%
%where $ h : \PowerSet{\indexset} \rightarrow \N$ is set as
%%
%\begin{equation}
%h\prl{I} \ldf  \card{I} - 1, \quad \forall I \subseteq \indexset. \label{eq.clustercardinalityheight}
%\end{equation}
%\end{definition}

\begin{lemma}\label{lem.CCembedding}
	The map $\mat{U}$ is injective.
\end{lemma} 
\begin{proof}
To see the injectivity of $\mat{U}$  (\refdef{def.UltrametricRepn}), we shall show that $\mat{U}\prl{\treeA} \neq \mat{U}\prl{\treeB}$ for any  $\treeA \neq \treeB \in \treetopspace_{\indexset}$. 

Two trees $\treeA, \treeB \in \treetopspace_{\indexset}$ are distinct if and only if they have at least one unshared cluster.
Accordingly, for any $\treeA \neq \treeB \in \treetopspace_{\indexset}$ consider a common cluster $I \in \cluster{\treeA} \cap \cluster{\treeB}$ with distinct parents $\parentCL{I,\treeA} \neq \parentCL{I,\treeB}$. 
Depending on the cardinality of parent clusters: 

\begin{itemize}

\item If $\card{\parentCL{I,\treeA}} = \card{\parentCL{I,\treeB}}$, then observe that there  exists some $j \in \parentCL{I,\treeA}$ such that  $j \not \in \parentCL{I,\treeB}$ because $\parentCL{I,\treeA} \neq \parentCL{I,\treeB}$. 
In fact, notice that  $j \in \complementLCL{ I}{\treeA}$ and $j \not \in \complementLCL{ I}{\treeB}$ (recall  \refeqn{eq:compLCL}). 
Hence, for any $i \in I$ we have  $\cancCL{i}{j}{\treeA} = \parentCL{I,\treeA}$ and $\parentCL{I,\treeB} \subsetneq \cancCL{i}{j}{\treeB}$. 
Thus, it follows from \refdef{def.UltrametricRepn} that  for any $i \in I$
\begin{equation}
\mat{U}\prl{\treeA}_{ij} = \card{\big.\parentCL{I,\treeA}} -1 < \mat{U}\prl{\treeB}_{ij} = \card{\big. \cancCL{i}{j}{\treeB}} -1 ~.
\end{equation}  

\item Otherwise, without loss of generality, let $\card{\parentCL{I,\treeA}} < \card{\parentCL{I,\treeB}}$. 
Then, observe that for any $i \in I$ and $j \in \complementLCL{I}{\treeA}$,
\begin{equation}
\mat{U}\prl{\treeA}_{ij} = \card{\parentCL{I,\treeA}} - 1 < \mat{U}\prl{\treeB}_{ij} = \card{\big. \cancCL{i}{j}{\treeB}} - 1 ~,
\end{equation}
since $\cancCL{i}{j}{\treeB} \supseteq \parentCL{I,\treeB}$.

\end{itemize}

Therefore, for any $\treeA \neq \treeB \in \bintreetopspace_{\indexset}$ one has $\mat{U}\prl{\treeA} \neq \mat{U}\prl{\treeB}$, and the result follows.
\end{proof}

Using the embedding $\mat{U}$  of $\treetopspace_{\indexset}$ into $\R^{\card{\indexset} \times \card{\indexset}}$, we can construct tree metrics by pulling back metrics induced from matrix norms, such as the one below:
\begin{definition} \label{def:ClusterCardinalityDist}
The \emph{cluster-cardinality metric}, $\distCC:\treetopspace_{\indexset} \times \treetopspace_{\indexset} \rightarrow \R_{\geq 0}$, on $\treetopspace_{\indexset}$ is defined to be
\footnote{Here $\norm{.}_1$ denotes the 1-norm of a matrix, i.e. $ \norm{\mat{U}}_1 \ldf \sum_{i=1}^{n} \sum_{j=1}^{n} \absval{\mat{U}_{ij}}$ for  $\mat{U} \in \R^{n \times n}$. Our choice of the 1-norm was guided by the resulting relationships between $\distCC$ and the dissimilarity measures $\distCM$ and $\distNav$ introduced below. Other choices of norm on $\R^{\indexset\times\indexset}$ may prove useful.}
\begin{equation}\label{eq.CCdistance}
\distCC\prl{\treeA,\treeB} \ldf \frac{1}{2}\norm{\big. \mat{U}\prl{\treeA} - \mat{U}\prl{\treeB}}_{1},  \quad \treeA, \treeB \in \treetopspace_{\indexset} ~.
\end{equation} 
\end{definition} 
\begin{proposition}\label{prop:CCdistanceCost}
The cluster-cardinality distance $\distCC$ on $\treetopspace_{\indexset}$ is computable in $\bigO{\card{\indexset}^2}$ time. 
\end{proposition}
\begin{proof} The 1-norm of the difference of a pair of $\card{\indexset} \times \card{\indexset}$ matrices obviously requires $\bigO{\card{\indexset}^2}$ time to compute, giving a lower bound on the computation cost of $\distCC$. It remains to show that the embedding $\mat{U}$ (\refdef{def.UltrametricRepn}) may be obtained at this cost. We proceed by induction based on a post-order traversal of the trees involved, $\tree \in \treetopspace_{\indexset}$. For the base case, consider the two-leaf tree $\tree\in \bintreetopspace_{\brl{2}}$, i.e.  $\card{\indexset} = 2$: then we simply assign $\mat{U}\prl{\tree} = \left[\begin{smallmatrix} 0&1\\ 1&0\end{smallmatrix}\right]$. For the induction step, assume $\card{S}\geq 3$ and denote $\childCL{\indexset,\tree}=\crl{\indexset_k}_{1 \leq k \leq K}$, where $K \geq 2$ is the number of children of the root $\indexset$ in $\tree$. We observe:
\begin{itemize}
	\item For every singleton child $\{i\}$ of $S$ in $\tree$ (if any), then set $\mat{U}\prl{\tree}_{ii}= 0$, which takes up $\bigO{1}$ time.
%Note that, given its children, one can obtain a cluster of each vertex of $\tree$  and its size in linear time during the post-order traversal of $\tree$ using $\indexset = \bigcup_{k =1}^{K} \indexset_k $ and $\card{\indexset} = \sum_{k=1}^{K}\card{\indexset_k}$. 
	\item Note that all clusters of $\tree$ and their sizes can be obtained in $\bigO{\card{\indexset}^2}$ time by a single post-order traversal, as each individual cluster (as well as its cardinality) takes at most linear time to compute from those of its children.
	\item Suppose that for any $1 \leq k \leq K$ and $\card{\indexset_k}\geq2$ the elements of $\mat{U}\prl{\tree}$ associated with the subtree rooted at $\indexset_k$ can be computed in $\bigO{\card{\indexset_k}^2}$ time. Then, the total number of  updates associated with the root $\indexset$ is $\sum_{k = 1}^{K} \sum_{l =1}^K \card{\indexset_k}\card{\indexset_l} $  and  corresponds to setting $\mat{U}\prl{\tree}_{ij} =\mat{U}\prl{\tree}_{ji}= \card{\indexset} - 1$ for all $i \in \indexset_k$, $ j \in \indexset_l$ and $1 \leq k,l \leq K$.
\end{itemize}
In total, the cost of obtaining $\mat{U}\prl{\tree}$ is $\sum_{k =1}^{K}\bigO{\card{\indexset_k}^2}  + \sum_{k = 1}^{K} \sum_{l=1}^{K} \card{\indexset_k}\card{\indexset_l} + \bigO{\card{\indexset}^2}= \bigO{\card{\indexset}^2}$, as required.% since $\sum_{k =1}^{K}\card{\indexset_k}^2  + \sum_{k = 1}^{K} \sum_{l=1}^{K} \card{\indexset_k}\card{\indexset_l} = \prl{\sum_{k =1}^{K} \card{\indexset_k}}^2 = \card{\indexset}^2 $. 
\end{proof}

The diameter, $\diam{X,\dist} \ldf \max \crl{\dist\prl{x,y} \big | \, x,y \in X}$, of a finite metric space $\prl{X,\dist}$ is always of interest in algorithmic applications. Some known diameters for hierarchies~\cite{bogdanowicz_giaro_TCBB2012, lin_rajan_moret_TCBB2012, sleator_tarjan_thurston_stoc1986} are:
\begin{equation}\label{eq:diameters}
\diam{\treetopspace_{\indexset}, \distRF}=\card{\indexset} - 2\,,\quad
\diam{\bintreetopspace_{\indexset}, \distMS}=\bigO{\card{\indexset}^2}\,,\quad
\diam{\bintreetopspace_{\indexset}, \distNNI}=\bigO{\card{\indexset}\log \card{\indexset}}
\end{equation}
For the cluster-cardinality distance we have:
\begin{proposition}\label{prop.CC_diameter}
$\diam{\treetopspace_{\indexset},\distCC} = \bigO{\card{\indexset}^3}$ .
\end{proposition}
\begin{proof} 
From \refdef{def.UltrametricRepn}, the minimum and maximum ultrametric distances between two distinct elements of $\indexset$ are, respectively,  1 and $\card{\indexset} - 1$, implying the bound
\begin{equation}\label{eq.PtwiseDifference}
	\max_{i,j \in \indexset} \prl{\mat{U}\prl{\treeA}_{ij}-\mat{U}\prl{\treeB}_{ij}}\leq \card{\indexset} - 2 ~.
\end{equation}
Moreover, using the tight upper bound on the change of the cluster-cardinality distance after a single NNI move from \refprop{prop:NNIClusterCardinality}, the diameter of $\treetopspace_{\indexset}$ with respect to $\distCC$ satisfies
\begin{equation}
\left \lfloor \frac{2}{27}\card{\indexset}^3 \right \rfloor \leq \diam{\treetopspace_{\indexset},\distCC} \leq \frac{1}{2}\card{\indexset}\prl{\card{\indexset}-1}\prl{\card{\indexset}-2},
\end{equation} 
which completes the proof.
\end{proof}

A common question regarding any distance being proposed for the space of trees is how it behaves with respect to certain tree rearrangements. 
For instance, any pair of NNI-adjacent trees, $\treeA, \treeB \in \bintreetopspace_{\indexset}$, are known to satisfy \cite{bogdanowicz_giaro_TCBB2012} \footnote{$\lfloor.\rfloor$ denotes the floor operator returning the largest integer not greater than its operand.} 

\noindent
\begin{align}
&\distNNI\prl{\treeA,\treeB} = 1 \Longleftrightarrow \distRF\prl{\treeA,\treeB} = 1 ~, \label{eq:dRF vs dNNI}\\
&\distNNI\prl{\treeA,\treeB} = 1 \Longrightarrow  2 \leq \distMS\prl{\treeA,\treeB} \leq \left \lfloor \frac{\card{\indexset}}{2} \right \rfloor ~.\label{eq:dMS vs dNNI}
\end{align} 
Similarly for $\distCC$ we have:
\begin{proposition} \label{prop:NNIClusterCardinality}
Let $\prl{\treeA,\treeB}$ be an edge of the NNI-graph $\NNIgraph_{\indexset} = \prl{\bintreetopspace_{\indexset}, \NNIedgeset}$ and $\prl{A,B,C}$ be the associated NNI triplet (\reflem{lem.NNITriple}). 
Then 
\begin{equation}\label{eq.CC_NNI_range}
2 \leq \distCC\prl{\treeA,\treeB} =  2\card{A}\card{B}\card{C} \leq \left \lfloor \frac{2}{27}\card{\indexset}^3 \right \rfloor ~,
\end{equation}
and both bounds are tight.
\end{proposition}
\begin{proof}
Let $P = A \cup B \cup C$ and recall from \reflem{lem.NNITriple} that $A \cup B \in \cluster{\treeA}$ and $B \cup C \in \cluster{\treeB}$. 
Note that $P \in \cluster{\treeA} \cap \cluster{\treeB}$ is a common (grand)parent cluster, and $A$, $B$ and $C$ are pairwise disjoint.  

Since the NNI moves between $\treeA$ and $\treeB$ only change the relative relations of clusters $A,B$ and $C$, the distance between $\treeA$ and $\treeB$ can be rewritten as 

\noindent
\begin{align}
\distCC\prl{\treeA,\treeB} &= \frac{1}{2}\norm{\big.\mat{U}\prl{\treeA} - \mat{U}\prl{\treeB}}_1 ~, \\
&= \sum_{\substack{i\in A \\ j \in B}} \absval{\mat{U}\prl{\treeA}_{ij} - \mat{U}\prl{\treeB}_{ij}} +  \sum_{\substack{i\in A \\ j \in C}} \absval{\mat{U}\prl{\treeA}_{ij} - \mat{U}\prl{\treeB}_{ij}}  +  \sum_{\substack{i\in B \\ j \in C}} \absval{\mat{U}\prl{\treeA}_{ij} - \mat{U}\prl{\treeB}_{ij}} ~, \\
& =  \sum_{\substack{i\in A \\ j \in B}} \underbrace{\absval{ h\prl{A \cup B} - h\prl{P}}}_{=\card{C}}  + \sum_{\substack{i\in A \\ j \in C}} \underbrace{\absval{h\prl{P} - h\prl{P}}}_{=0} +  \sum_{\substack{i\in B \\ j \in C}} \underbrace{\absval{h\prl{P} - h \prl{B \cup C} }}_{=\card{A}} ~, \\
& =  2\card{A}\card{B}\card{C}.
\end{align}
Clearly, the lower bound in \refeqn{eq.CC_NNI_range} is realized when $\card{A}=\card{B}=\card{C}=1$. 
Since the maximum product of three numbers with a prescribed sum occurs when all the numbers are equal~ ---~ in our case, $\card{A}+\card{B} + \card{C} \leq \card{\indexset}$~ ---~ we must have $\card{A}\card{B} \card{C} \leq \left \lfloor\frac{\card{\indexset}^3}{27} \right \rfloor$, as $\card{\big. .}$ is integer-valued. 
The result follows. 
\end{proof}

Inequalities of the above form allow one to take advantage of the combinatorial nature of $\distNNI$ through repeated application of the triangle inequality:
\begin{corollary} \label{cor:NNI_RF_bound}
Over $\bintreetopspace_{\indexset}$ one has $\distRF \leq \distNNI$.
\end{corollary}
Indeed, the length of a path in $\NNIgraph_{\indexset}$ produces a bound on the RF distance between its endpoints by repeatedly applying the triangle inequality to \eqref{eq:dRF vs dNNI}. A similar argument yields:
\begin{corollary}\label{cor:NotMetric}
Let $\dist$ be a dissimilarity on $\bintreetopspace_{\indexset}$ with the property that $\dist\prl{\treeA,\treeB} \leq 1$ for any pair of NNI-adjacent hierarchies $\treeA,\treeB \in \bintreetopspace_{\indexset}$.
If $\dist\prl{\treeA,\treeB}> \distNNI\prl{\treeA,\treeB}$ for some $\treeA,\treeB \in \bintreetopspace_{\indexset}$, then $d$ is not a metric. 
\end{corollary}

%\refcor{cor:NotMetric} immediately applies to the dissimilarity measures $\distCM$ and $\distNav$, defined below, by observing the relation in \reffig{fig:CMatNormNavLengthNotMetric}: 
%%
%\begin{corollary} \label{cor.CMdistanceNotMetric}
%The crossing dissimilarity, $\distCM$ \refeqn{eq.CMdistance}, and NNI navigation dissimilarity, $\distNav$ \refeqn{eq:NavPathLength},  are not metrics.
%\end{corollary}

%%%%%%%%%%%%%%%%%%%%%%%%%%%%%%%%%%%%%%%%%%%
%%%%%%%%%%%%%%%%%%%%%%%%%%%%%%%%%%%%%%%%%%%
\subsection{The Crossing Dissimilarity}
\label{sec.CMdistance}
%%%%%%%%%%%%%%%%%%%%%%%%%%%%%%%%%%%%%%%%%%%
%%%%%%%%%%%%%%%%%%%%%%%%%%%%%%%%%%%%%%%%%%%

\begin{definition}\label{def.CompatibilityAndCrossingMat}
Let $\treeA,\treeB \in \treetopspace_{\indexset}$.
We define their \emph{compatibility matrix} $\CMat\prl{\treeA,\treeB}$ and their \emph{crossing matrix} $\XMat\prl{\treeA,\treeB}$ to be\footnote{$\CMat\prl{\treeA,\treeB}$ and $\XMat\prl{\treeA, \treeB}$ can be defined only in terms of nontrivial clusters of $\treeA$ and $\treeB$ since  any trivial cluster of $\treeA$ and $\treeB$ is compatible with any cluster $K \subseteq \indexset$. As a result, we are required to separately consider the special case in which one of the trees has only trivial clusters whenever $\CMat$ or $\XMat$ are used to reason about degenerate trees.}
\begin{equation}\label{eq.CMatXMat}
	\CMat\prl{\treeA,\treeB}_{I,J} \ldf \indicator\prl{I \compatible J} \quad %\forall I \in \cluster{\treeA}, J \in \cluster{\treeB},\quad
	\text{ and } \quad
	\mat{X}\prl{\treeA,\treeB}_{I,J} \ldf 1 - \CMat\prl{\treeA,\treeB}_{I,J},  
\end{equation}
where $I \in \cluster{\treeA}, J \in \cluster{\treeB}$ and $\indicator\prl{.}$ denotes the indicator function returning unity if its argument holds true and zero otherwise. The \emph{crossing dissimilarity} $\distCM$ is defined by
$\distCM\prl{\treeA,\treeB} \ldf \norm{\big.\XMat\prl{\treeA,\treeB}}_1\,$, counting\footnote{We find that choosing to use the 1-norm of the crossing matrix easily reveals combinatorial relations between $\dist_{CM}$ and $\distCC$ \refeqn{eq.CCdistance}; of course, one could use other matrix norms to construct alternative dissimilarities.} the pairs of incompatible clusters in $\cluster{\treeA}\cup\cluster{\treeB}$.
\end{definition}

We list some useful properties of $\distCM$:
\begin{remark} \label{rem.distCMNotMetric}
The crossing  dissimilarity $\distCM$ on $\bintreetopspace_{\indexset}$ is positive definite and  symmetric, but it is not a metric (apply \refcor{cor:NotMetric} to the observations of  \reffig{fig:CMatNormNavLengthNotMetric}).
\end{remark} 

\medskip

\begin{figure}[htb]
\centering
\includegraphics[width=0.60 \textwidth]{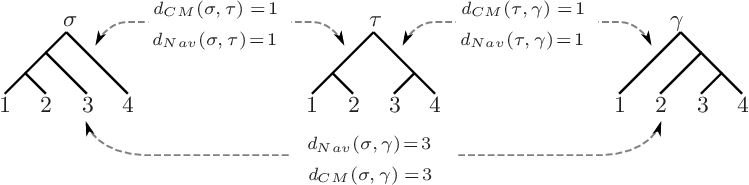}
%\vspace{-2mm} 
\caption{ $\distCM$ and $\distNav$ are not metrics: an example of the triangle inequality failing for both dissimilarities.}
\label{fig:CMatNormNavLengthNotMetric}
\end{figure}

\smallskip
 
\begin{proposition}\label{prop:CMdistanceCost}
The crossing dissimilarity $\distCM$ over  $\treetopspace_{\indexset}$ can be computed in $\bigO{\card{\indexset}^2}$ time. 
\end{proposition} 
\begin{proof}
The crossing matrix $\XMat\prl{\treeA,\treeB}$ \refeqn{eq.CMatXMat} of a pair of hierarchies  $\treeA, \treeB \in \treetopspace_{\indexset}$ has at most $2\card{\indexset} -1$ rows and columns.
Hence, the 1-norm of $\XMat\prl{\treeA,\treeB}$ requires $\bigO{\card{\indexset}^2}$ time to compute, bounding the cost of $\dist_{CM}$ from below.
To obtain the upper bound, we show that $\XMat\prl{\treeA, \treeB}$ can be obtained in $\bigO{\card{\indexset}^2}$ time by post-order traversal.

Observe that for any cluster $J \in \cluster{\treeB}$ (and symmetrically, for any cluster of $\cluster{\treeA}$) one can check whether $J$ is disjoint with or a superset of each cluster $I$ of $\treeA$ by a post-order traversal of $\treeA$ in $\bigO{\card{\indexset}}$ time using the following recursion: 
\begin{itemize}
	\item If either $I$ or $J$ is a singleton then the cluster inclusions $I \subseteq J$, $J \subseteq I$ and their disjointness can be determined in constant time using a hash map.
	\item Otherwise ($\card{I} \geq 2$ and $\card{J} \geq 2$), we have
	\begin{align}
		I \subseteq J & \Longleftrightarrow \forall D \in \childCL{I, \treeA} \quad D \subseteq J , \\
		I \cap J = \emptyset &\Longleftrightarrow \forall D \in \childCL{I,\treeA} \quad D \cap J = \varnothing. 
	\end{align}
\end{itemize}

Thus, it follows from \refdef{def.compatibility} that a complete list of compatibilities between $\treeA$ and $\treeB$ can be produced in $\bigO{\card{\indexset}^2}$ time, and so  $\XMat\prl{\treeA,\treeB}$ can be obtained at the same cost, $\bigO{\card{\indexset}^2}$.
\end{proof}
 
\begin{proposition} \label{prop.CMdiameter}
$\diam{\treetopspace_{\indexset}, \distCM} = \prl{\card{\indexset} - 2}^2$ . 
%The diameter $\diam{\treetopspace_{\indexset},\distCM}$ \refeqn{eq.diameter} of the set of hierarchies over a fixed finite index set $\indexset$ with respect to the crossing dissimilarity $\distCM$ \refeqn{eq.CMdistance} is 
%%
%\begin{equation} \label{eq.CMdiameter}
%\diam{\treetopspace_{\indexset}, \distCM} = \prl{\card{\indexset} - 2}^2. 
%\end{equation}
\end{proposition}
\begin{proof}
Two clusters of a pair of trees can only be incompatible if they are both nontrivial.
Recall from \refrem{rem:MaximumCardinality} that  the number of nontrivial clusters of a tree  in $\treetopspace_{\indexset}$ is at most $\card{\indexset} -2$. 
Hence, by \refdef{def.CompatibilityAndCrossingMat}, an upper bound on $\diam{\treetopspace_{\indexset}, \distCM}$ is $\,\prl{\card{\indexset} -2}^2$. 
To observe that this upper bound is realized, see  \reffig{fig.distCMdistNavDiameter}.
\end{proof} 

\smallskip

\begin{figure}[htb] 
\centering
\includegraphics[width=0.5\textwidth]{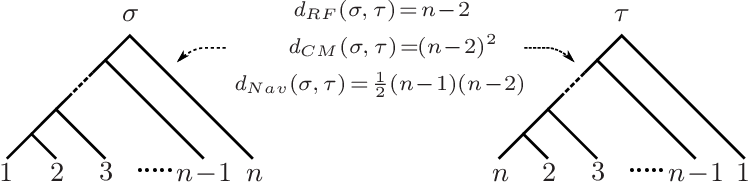} 
\caption{A pair of nondegenerate hierarchies realizing  
 $\diam{\treetopspace_{\brl{n}}, \distCM} = \prl{n - 2}^2$ and 
 $\diam{\bintreetopspace_{\brl{n}}, \distNav} = \frac{1}{2}\prl{n - 1}\prl{n -2}$.}
\label{fig.distCMdistNavDiameter}
\end{figure}
 
\begin{proposition} \label{prop.NNI_CM_neighbor}
Two nondegenerate trees $\treeA,\treeB \in \bintreetopspace_{\indexset}$ are NNI-adjacent if and only if $\distCM\prl{\treeA,\treeB} = 1$. 
\end{proposition}
\begin{proof}
The result is  evident from \refrem{rem:MaximumCardinality} and \refdef{def:NNIMove}.
\end{proof} 

\noindent Despite the result of the last proposition, $\distCM$ does not provide a linear lower bound on $\distNNI$ since $\diam{\bintreetopspace_{\indexset}, \dist_{NNI}} = \bigO{\card{\indexset} \log \card{\indexset}} < \diam{\bintreetopspace_{\indexset}, \dist_{CM}} = \bigO{\card{\indexset}^2}$ (\refprop{prop.CMdiameter}).
This inequality provides us with an additional, more conceptual, argument that $\distCM$ is not a metric, by applying \refcor{cor:NotMetric}.

\begin{proposition} \label{prop:CM_RF_bound}
%The crossing dissimilarity over $\bintreetopspace_{\indexset}$ is bounded by the Robinson-Foulds distance as follows:
Over $\treetopspace_{\indexset}$ one has $\distRF \leq \dist_{CM} \leq \distRF^2\,$. 
%
%\begin{equation} \label{eq.CM_RF_bounds}
%\distRF\prl{\treeA, \treeB} \leq \dist_{CM}\prl{\treeA, \treeB} \leq \distRF\prl{\treeA, \treeB}^2.
%\end{equation}
These bounds are tight.
\end{proposition}
\begin{proof}
The lower bound directly follows from \refrem{rem:MaximumCardinality}. 
Because a pair of distinct binary hierarchies always have uncommon clusters whose count is equal to $\distRF$, and an unshared cluster of one tree crosses at least one unshared cluster of the other tree.
This bound is  tight since for any $\treeA, \treeB \in \bintreetopspace_{\indexset}$
\begin{equation}
\distRF\prl{\treeA, \treeB} = 1 
\Leftrightarrow 
\distNNI\prl{\treeA, \treeB} = 1 
\Leftrightarrow 
\distCM\prl{\treeA, \treeB} = 1.
\end{equation}

For any $\treeA, \treeB \in \bintreetopspace_{\indexset}$, the columns and rows of $\XMat\prl{\treeA, \treeB}$ \refeqn{eq.CMatXMat} associated with common clusters of $\treeA,\treeB$ are necessarily null.
Hence, $\XMat\prl{\treeA, \treeB}_{I,J}\neq 0$ implies $I\notin\cluster{\treeB}$ and $J\notin\cluster{\treeA}$. By the definition of $\distRF$, there are no more than $\distRF\prl{\treeA, \treeB}^2$ such pairs~ ---~ hence the claimed upper bound. To observe that this bound is also tight, see  \reffig{fig.distCMdistNavDiameter}.
\end{proof}

\begin{proposition} \label{prop:CM_CC_bound}
%The crossing dissimilarity is  bounded from above  by the cluster-cardinality distance, i.e. 
%%
%\begin{equation}
%\distCM\prl{\treeA,\treeB} \leq \distCC\prl{\treeA,\treeB} \quad \forall \treeA,\treeB \in \treetopspace_{\indexset}.
%\end{equation}
Over $\treetopspace_{\indexset}$ one has $\dist_{CM} \leq \distCC\,$.
\end{proposition} 
\begin{proof}
Given any $\treeA,\treeB \in \treetopspace_{\indexset}$ we claim that there is a function $q : \cluster{\treeA} \times \cluster{\treeB} \rightarrow S \times S$ with the following properties:
\begin{enumerate}[label=(\alph*)]
\item \label{it.q1}  for any  $I \in \cluster{\treeA}$ and $J \in \cluster{\treeB}$, $ I \compatible J$ if and only if  $\,\prl{i,j} = q\prl{I,J}$ with $i = j$, 
\item \label{it.q2} for any $i \neq j \in \indexset$,  $\card{q^{-1}\prl{i,j}} \leq \absval{\mat{U}\prl{\treeA}_{ij} - \mat{U}\prl{\treeB}_{ij}}$.  
\end{enumerate}

\noindent Observe that, if such a function does exist, then \ref{it.q1} implies:
\begin{equation} \label{eq.crossingsqinverse}
\bigcup_{i \neq j \in S} q^{-1}\prl{i,j} =\crl{\,\prl{I,J}  \in \cluster{\treeA} \times \cluster{\treeB} \Big | \, I \not \compatible J}.
\end{equation}
It is then evident from  \refeqn{eq.crossingsqinverse} and \ref{it.q2} that
\begin{equation}
\distCM\prl{\treeA, \treeB} \leq \sum_{i \neq j \in \indexset} \card{q^{-1}\prl{i,j}} \leq \distCC\prl{\treeA,\treeB},
\end{equation} 
proving our proposition.

We proceed to construct the function $q$. If $I \not \compatible J$, then  there exist $i \in I \cap J$  and $j \in I \setminus J$ with the property that $\cancCL{i}{j}{\treeA} = I$.
Accordingly, define 

\noindent
\begin{align}
Q\prl{I,J} &\ldf \crl{\, \prl{i,j} \in \indexset \times \indexset  \Big | \, i \in I \cap J,  j \in I \setminus J, \, \cancCL{i}{j}{\treeA} = I },  \label{eq:Qset}
\\
R\prl{I,J} & \ldf \crl{\, \prl{i,j} \in \indexset \times \indexset \Big | \, i \in I \cap J,  j \in J \setminus I, \, \cancCL{i}{j}{\treeB} = J }. 
\end{align}
Note that if $\prl{i,j}  \in Q\prl{I, J} \cup R\prl{I,J}$, then $i \neq j$.

Have $\indexset$ totally ordered (say, by enumerating its elements) and have $\indexset \times \indexset$ ordered lexicographically according to the order of $\indexset$.
Then,  define $q : \cluster{\treeA} \times  \cluster{\treeB}  \rightarrow \indexset \times \indexset$ to be 
\begin{equation}
q\prl{I,J} \ldf 
\left \{ \begin{array}{cl}
\prl{ \big. \min\prl{I \cup J}, \min\prl{I \cup J}} & \text{, if } I \compatible J,\\
\min Q\prl{I,J} & \text{, if } I \not \compatible J, \card{I} \leq \card{J}, \\
\min R\prl{I,J} & \text{, if } I \not \compatible J, \card{I} > \card{J}. 
\end{array} \right . 
\end{equation}
Recall that  $Q\prl{I,J}$ and $R\prl{I,J}$ both contain pairs of distinct elements of $\indexset$. 
Hence,  $q$ satisfies the property \ref{it.q1} above.
 
By construction, for any $i\neq j$ we have:
\begin{equation} \label{eq.qinverse}
q^{-1}\prl{i,j} \subseteq A\prl{i,j} \cup B\prl{i,j},
\end{equation}
where

\noindent
\begin{align}
A \prl{i,j} &\ldf  \left \{\, \prl{I,J} \in \cluster{\treeA} \times \cluster{\treeB}  \Big| \, I \not \compatible J,  \card{I} \leq \card{J}, \,   \prl{i,j} \in Q\prl{I,J} \right \},
\\
B\prl{i,j} &\ldf \left \{ \, \prl{I,J} \in \cluster{\treeA} \times \cluster{\treeB}  \Big  | \, I \not \compatible J,  \card{I} \geq \card{J}, \, \prl{i,j} \in R\prl{I, J} \right \}.
\end{align}
Remark from \refeqn{eq:Qset} that if $\prl{I,J} \in A\prl{i,j}$ then $\cancCL{i}{j}{\treeA} = I$ and $\cancCL{i}{j}{\treeB} \supsetneq J$. 
Hence,  if $\card{\,\cancCL{i}{j}{\treeA}} \geq \card{\,\cancCL{i}{j}{\treeB}}$, then $A\prl{i,j} = \varnothing$. 
Similarly, $\cancCL{i}{j}{\treeA} \supsetneq I$ and $\cancCL{i}{j}{\treeB} = J$ whenever $\prl{I,J} \in B\prl{i,j}$; and $B\prl{i,j} = \varnothing$ if $\card{\,\cancCL{i}{j}{\treeA}} \leq \card{\,\cancCL{i}{j}{\treeB}}$.
Thus, one can observe that   for any $i,j \in \indexset$,
\begin{equation}
A\prl{i,j} \neq \varnothing \Longrightarrow B\prl{i,j} = \varnothing. \label{eq.incompatiblepairs}
\end{equation}

Recall that for any $i,j \in \indexset$ and $\prl{I,J} \in A\prl{i,j}$ we have:
\begin{equation}
I = \cancCL{i}{j}{\treeA}, \, J \subsetneq \cancCL{i}{j}{\treeB}, \,  \card{I}  \leq \card{J} \text{ and } J \in \ancestorCL{\crl{i}, \treeB}.
\end{equation}
Hence,  one can conclude that
\begin{equation}
\card{\big.A\prl{i,j}} \leq \absval{\Big.\card{\,\cancCL{i}{j}{\treeB}} - \card{\,\cancCL{i}{j}{\treeA}}}
= \absval{\mat{U}\prl{\treeB}_{ij} - \mat{U}\prl{\treeA}_{ij}}. 
\end{equation}    
Similarly, for any $i,j \in \indexset$
\begin{equation}
\card{\big.B\prl{i,j}} \leq \absval{\Big. \card{\,\cancCL{i}{j}{\treeA}} - \card{\,\cancCL{i}{j}{\treeB}}}
= \absval{\mat{U}\prl{\treeA}_{ij} - \mat{U}\prl{\treeB}_{ij}}.  
\end{equation}

Thus, overall, using \refeqn{eq.qinverse} and \refeqn{eq.incompatiblepairs}, one can obtain the second property of $q$ as follows: for any $i \neq j \in S$
\begin{equation}
\card{q_{\treeA, \treeB}^{-1}\prl{i,j}} \leq \card{\big.A\prl{i,j}} + \card{\big.B\prl{i,j}} \leq \absval{\mat{U}\prl{\treeB}_{ij} - \mat{U}\prl{\treeA}_{ij}},
\end{equation}
which completes the proof.
\end{proof}

%%%%%%%%%%%%%%%%%%%%%%%%%%%%%%%%%%%%%%%%
%%%%%%%%%%%%%%%%%%%%%%%%%%%%%%%%%%%%%%%%
\section{The Navigation Dissimilarity}
\label{sec.dnav}
%%%%%%%%%%%%%%%%%%%%%%%%%%%%%%%%%%%%%%%%
%%%%%%%%%%%%%%%%%%%%%%%%%%%%%%%%%%%%%%%%

\refprob{navigation problem} may be loosely restated in graph-theoretic terms as follows:
\begin{problem} For each tree $\treeB\in\BT{\indexset}$, find a subgraph $\NNIdir{\indexset,\treeB}$ of the NNI graph $\NNIgraph_{\indexset}$ containing no directed cycles and such that every $\treeA\in\BT{\indexset}$ satisfies:
\begin{itemize}
	\item[($\dagger$)] If $\treeA\neq\treeB$ then there exists an edge of $\NNIdir{\indexset, \treeB}$ exiting $\treeA$; moreover, such an edge may be produced in low time complexity.
\end{itemize}
\end{problem}
Clearly, the reactive navigation algorithm $\mathcal{A}_\treeB$ of \refprob{navigation problem} is, in this case, to compute an edge of $\NNIdir{\indexset,\treeB}$ exiting the input tree $\treeA$ and then follow that edge. The challenge for us is to produce a graph (\refdef{def:navigation to a point})  where (i) the complexity of $\mathcal{A}_\treeB$ is low (\refcor{cor:NNINavMoveCost}), and (ii) the length of any directed path is bounded by a reasonable function of $\distNNI\prl{\treeA,\treeB}$, or, at least of $n=\card{\indexset}$ (\refdef{def:navigation distance}, \refthm{thm:navigation graphs work} and \refcor{cor:NNINavComplexity}).
Observe the similarity between our requirements of $\NNIdir{\indexset,\treeB}$ and a skeletal variant of the stricter notion of a {\it combing} from the early days of geometric group theory (see, e.g. \cite{word_processing_in_groups}): a `coherent' system of paths $\{p_x\}_x\in X$ in a topological space $X$, one for each point of the space, with $p_x(0)=x_0$ for all $x\in X$ and $p_y(t)=p_x(t)$ for all $t\leq s$ whenever $y=p_x(s)$. Specializing to the differentiable setting, one might hope to be able to (efficiently) compute a tangent vector $t_x$ to $p_x$ at $x$ in some open dense (and necessarily contractible) sub-manifold of $X$ so that the $p_x$ become integral curves of $\dot x=t_x$; following these curves in reverse comprises reactive navigation towards $x_0$, as seen through the eyes of a roboticist \cite{burridge_rizzi_kod_IJRR1999}.

\bigskip
We start out with a study of the coarse structure of the directed NNI graph $\NNIgraph_{\indexset}$. We consider special subspaces of the vertex space $\BT{\indexset}$:
\begin{definition}\label{def:Subspaces}
Let $K_1,\ldots,K_m$, $m\geq 1$, be a compatible family of subsets of $\indexset$. Denote:
\begin{equation}\label{eq:subspace}
	\BT{\indexset}(K_1,\ldots,K_m)\ldf\crl{
		\treeA\in\BT{\indexset}\,\big|\,
		\treeA\compatible\{K_1,\ldots,K_m\}
	}
\end{equation}
\end{definition}
Recalling that $\cluster{\treeA}$ is a maximal nested family in $\PowerSet{\indexset}$ if and only if $\treeA\in\BT{\indexset}$, one has, in fact:
\begin{equation}\label{eq:subspace2}
	\BT{\indexset}(K_1,\ldots,K_m)=\crl{
		\treeA\in\BT{\indexset}\,\big|\,
		K_1,\ldots,K_m\in\cluster{\treeA}
	}
\end{equation}
Intuitively, it is clear that the problem of navigating $\NNIdir{\indexset}$ towards a specified tree $\treeB$ may be parsed into a sequence of problems, each being that of navigating in $\BT{\indexset}(K)$ towards $\BT{\indexset}\prl{K} \cap \BT{\indexset}\prl{\childCL{K,\treeB}}$, where $K$ ranges over $\cluster{\treeB}$, starting with $K=\indexset$ and continuing inductively, provided each step preserves the achievements of its predecessors.

%%%%%%%%%%%%%%%%%%%%%%%%%%%%%%%%%%%%%%%%%%%
%%%%%%%%%%%%%%%%%%%%%%%%%%%%%%%%%%%%%%%%%%%
\subsection{Resolving incompatibilities with a prescribed split}
\label{sec:projectors} 
%%%%%%%%%%%%%%%%%%%%%%%%%%%%%%%%%%%%%%%%%%%
%%%%%%%%%%%%%%%%%%%%%%%%%%%%%%%%%%%%%%%%%%%

Throughout this section, let $\mathbb{K}=\{K_1,K_2\}$ be a fixed pair of disjoint non-empty subsets of $\indexset$, and set $K=K_1\cup K_2$. We will refer to such pairs as {\it partial splits}. Let us make a simple observation:
\begin{lemma}\label{lem:SplitCompatibleEquivalence}
The following equivalence holds for all $I\subsetneq K$:
\begin{equation} \label{eq:SplitCompatibleEquivalence}
	I\compatible\mathbb{K}\Longleftrightarrow
		\prl{I\subseteq K_1}\vee\prl{I\subseteq K_2}
\end{equation}
\end{lemma}
\begin{proof}
Suppose $I\compatible\mathbb{K}$ but neither $I\subseteq K_1$ nor $I\subseteq K_2$ holds. By \refdef{def.compatibility} we must then have $I\supseteq K_1$ and $I\supseteq K_2$, implying $I \supseteq K$~ ---~ contradiction to $I\subsetneq K$. The converse is trivial.
\end{proof}

\medskip
Let $\treeA\in\BT{\indexset}(K)$ be a tree which splits $K$ into a pair of children not coinciding with $\mathbb{K}$. According to the preceding lemma, this is equivalent to $\childCL{K,\treeA}\not\compatible\mathbb{K}$. Observe now that any cluster $I\in\cluster{\treeA}$ which is not a $\treeA$-descendant of $K$ is automatically compatible with $\mathbb{K}$. Thus, incompatibilities of $\treeA$ with $\mathbb{K}$ could only occur among $\treeA$-descendants of $K$. This motivates the following definition:
\begin{definition}[Recombinants] \label{def:Recombinants}
For $\treeA\in \BT{\indexset}(K)$ we distinguish two classes of $\treeA$-descendants of the cluster $K$:
\begin{eqnarray}
\CLXSplitSet{\treeA; \mathbb{K}} &\ldf&
	\left\{ \CLXSplit \in \descendantCL{K,\treeA}\,\big|\,
		\CLXSplit \not \compatible \mathbb{K}
	\right\},
\label{eq:Incompatibles}\\
\DeepXSplit{\treeA; \mathbb{K}} &\ldf&  
	\left\{ I \in \CLXSplitSet{\treeA; \mathbb{K}}\,\big|\,
		\childCL{I,\treeA} \compatible \mathbb{K}\,,\;
		\childCL{\complementLCL{I}{\treeA}, \treeA}
			\compatible \mathbb{K}\, 
	\right\} 
\label{eq:Recombinants}
\end{eqnarray}
For lack of a better term, we will refer to the elements of $\DeepXSplit{\treeA;\mathbb{K}}$ as \emph{recombinants} of $\mathbb{K}$ in $\treeA$. See \reffig{fig:CrossingClusters}.
\end{definition}
%Another immediate corollary of \reflem{lem:SplitCompatibleEquivalence} is then:
%\begin{corollary} \label{cor:AncestorIncompatibility}
%For any $\treeA\in\BT{\indexset}(K)$, if $I \in \CLXSplitSet{\treeA;K_1,K_2}$, then $\ancestorCL{I,\treeA} \cap \descendantCL{K,\treeA} \subseteq \CLXSplitSet{\treeA;K_1,K_2}$. 
%\end{corollary}

The set of recombinants suffices to characterize the compatibility of a tree with a given split:
\begin{lemma}\label{lem:recombinants_characterize_splits} Observe that $\treeA\in\BT{\indexset}(K)$ has recombinants of $\mathbb{K}$ if and only if $\treeA\notin\BT{\indexset}(\mathbb{K})$.
\end{lemma}
\begin{proof} Indeed, if $\treeA\in\BT{\indexset}(\mathbb{K})$, then all clusters of $\treeA$ are compatible with $\mathbb{K}$, causing $\CLXSplitSet{\treeA; \mathbb{K}}$ ~ ---~ and hence also $\DeepXSplit{\treeA; \mathbb{K}}$~ ---~ to be empty. Conversely, suppose there is a cluster of $\treeA$ incompatible with $\mathbb{K}$. Then the $\treeA$-children of any deepest such cluster  and  its local complement's children are compatible with $\mathbb{K}$ in $\treeA$, and their children are compatible with $\mathbb{K}$ as well (even if vacuously).
\end{proof}

\begin{figure}[htb]
\centering
\begin{tabular}{c@{\hspace{15mm}}c}
\raisebox{10mm}{\includegraphics[width=0.3\textwidth]{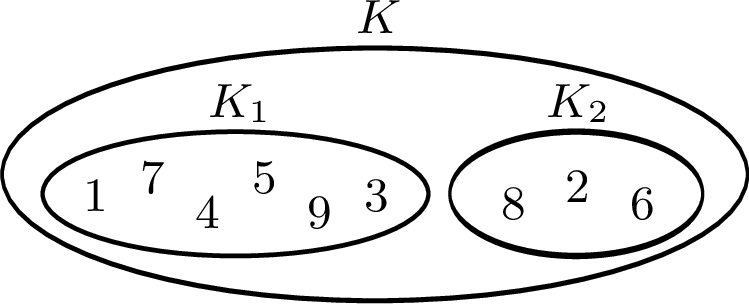}}
&
\includegraphics[width=0.375\textwidth]{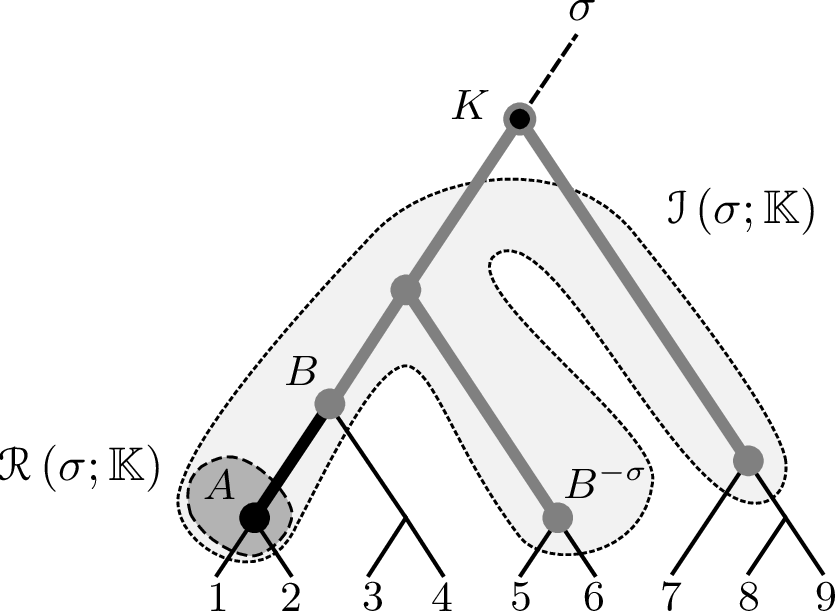}
\end{tabular}
\vspace{-2mm}
\caption{An illustration of $\CLXSplitSet{\treeA;\mathbb{K}}$ \refeqn{eq:Incompatibles} and $\DeepXSplit{\treeA;\mathbb{K}}$ \refeqn{eq:Recombinants} of $\treeA\in \BT{\brl{n}}(K)$, where $n\geq 9$ and $K = \brl{9}$. The vertices and edges associated with clusters of $\treeA$ incompatible with the split $\mathbb{K}$ are thickened.
The only recombinant of $\mathbb{K}$ in $\treeA$ is $A = \crl{1,2}$ and it has Type 1. 
$B$ and $\complementLCL{B}{\treeA}$ are examples of Type 2 clusters of $\treeA$ incompatible with $\mathbb{K}$ which are not recombinants. 
}  \label{fig:CrossingClusters}
\end{figure}

\begin{definition}[Incompatibility Types]\label{def:SingleJointIncompatibility} Given $\treeA\in\BT{\indexset}(K)$, a cluster $I\in\CLXSplitSet{\treeA;\mathbb{K}}$ is said to be of \emph{type 1} with respect to $\mathbb{K}$ if $\complementLCL{I}{\treeA} \compatible\mathbb{K}$. %The number of distinct such $I$ will be denoted by $\typeone{K_1,K_2}{\treeA}$. 
If $I\in\CLXSplitSet{\treeA;\mathbb{K}}$ is not of type 1, then it is said to be of \emph{type 2} (see \reffig{fig:CrossingClusters}). %The number of type 2 clusters will be denoted by $\typetwo{K_1,K_2}{\treeA}$.
\end{definition}
%An optimistic person will expect some function of $\typeone{K_1,K_2}{\treeA}$ and $\typetwo{K_1,K_2}{\treeB}$ to serve as a workable quantifier of the ``degree of incompatibility'' of a tree $\treeA$ with the split $\{K_1,K_2\}$. As we shall see in the following section, one successful candidate is:
%\begin{definition}\label{def:projection complexity} Let $\mathbb{K}=\{K_1,K_2\}$ be a partial split. For every $\treeA\in\BT{\indexset}(K_1\cup K_2)$ we define %$\projpaths{K_1,K_2}{\treeA}$ to be the set of maximal directed paths in $\projector{K_1,K_2}$ emanating from $\treeA$, and set
%\begin{equation}
%	\treenorm{\mathbb{K}}{\treeA}\ldf
%	\typeone{\mathbb{K}}{\treeA}+\tfrac{3}{2}\typetwo{\mathbb{K}}{\treeA}\,.
%\end{equation}
%which equals the common length of all paths in $\projpaths{K_1,K_2}{\treeA}$.
%\end{definition}

Another, perhaps less intuitive, quantifier of incompatibility arises as follows:
\begin{definition}[Essential Crossing Index]\label{def:crossing index} Let $\mathbb{K}=\{K_1,K_2\}$ and $\mathbb{L}=\{L_1,L_2\}$ be partial splits. Their \emph{essential crossing index} is defined as:
\begin{equation}
	\xindex{\mathbb{L}}{\mathbb{K}}\ldf\left\{\begin{array}{rl}
		0	&\text{if }\mathbb{L}\res{K_1\cup K_2}\compatible\mathbb{K}\res{L_1\cup L_2}\\
		1	&\text{if }L_j\res{K_1\cup K_2}\compatible\mathbb{K}\res{L_1\cup L_2}\text{ for only one }j\in\{1,2\}\\
		3	&\text{otherwise}
	\end{array}\right.
\end{equation}
For a tree $\treeA\in\BT{\indexset}$ we define:
\begin{equation}
	\crossnorm{\mathbb{K}}{\treeA}\ldf
	\sum_{I\in\cluster{\treeA}}\xindex{\childCL{I,\treeA}}{\mathbb{K}}
\end{equation}
\end{definition}
The following elementary observations will be useful:
\begin{lemma}\label{lem:xindex is symmetric} Let $\mathbb{K}=\{K_1,K_2\}$ and $\mathbb{L}=\{L_1,L_2\}$ be partial splits. Then $\xindex{\mathbb{K}}{\mathbb{L}}=\xindex{\mathbb{L}}{\mathbb{K}}$.
\end{lemma}
\begin{proof} 
Write $K=K_1\cup K_2$ and $L=L_1\cup L_2$. Without loss of generality we may assume $K=L=\indexset$, since:
\begin{equation}
	\bigcup_{i=1}^2(K_i\cap L)=\bigcup_{j=1}^2(L_j\cap K)=K\cap L\,.
\end{equation}
We study the possible cases:
\begin{itemize}
	\item{$\xindex{\mathbb{L}}{\mathbb{K}}=0$: } By definition, this means none of the $K_i$ crosses any of the $L_j$; equivalently, no $L_j$ crosses any of the $K_i$ and we have $\xindex{\mathbb{K}}{\mathbb{L}}=0$.
	\item{$\xindex{\mathbb{L}}{\mathbb{K}}=1$: } WLOG, only $L_1$ crosses $\mathbb{K}$, hence $L_2$ is contained in one of the $K_i$, say $K_2$. Then $L_1$ contains $K_1$ and at least one element of $K_2$, by \reflem{lem:SplitCompatibleEquivalence}. Thus, $K_1\compatible\mathbb{L}$ while $K_2\compatible L_2$, $K_2\not\compatible L_1$. This means $\xindex{\mathbb{K}}{\mathbb{L}}=1$.	  
	\item{$\xindex{\mathbb{L}}{\mathbb{K}}=3$: } if both $L_1$ and $L_2$ cross $\mathbb{K}$, then $L_j\cap K_i\neq\varnothing$ for all $i,j\in{1,2}$, implying both $K_1$ and $K_2$ cross $\mathbb{L}$, as desired. \qedhere
\end{itemize}
\end{proof}

We are now ready to construct the graph $\projector{\mathbb{K}}$:
\begin{definition}[Projector Graph] \label{def.Projector} Let $\mathbb{K}=\{K_1,K_2\}$ be a partial split, and set $K=K_1\cup K_2$. Then $\projector{\mathbb{K}}$ is defined to be the directed graph with vertex set $\BT{\indexset}(K)$, and all edges of the form $(\treeA,G)\in\NNIDedgeset$ such that $I\ldf\parentCL{G,\treeA}\in\DeepXSplit{\treeA;\mathbb{K}}$ and one of the following holds:
\begin{enumerate}
	\item \label{it.ProjectorType1} $I$ is of type 1, and $\complementLCL{G}{\treeA},\complementLCL{I}{\treeA}\subseteq K_i$ for some $i\in\{1,2\}$;
	\item \label{it.ProjectorType2} $I$ is of type 2.
\end{enumerate}
\end{definition}
\begin{figure}[htb]
\begin{center}
%\begin{tabular}{@{}c@{}}
%\begin{tabular}{@{}c@{\hspace{5mm}}}
%\includegraphics[width=0.125\textwidth]{figures/CommonXCLv2.eps} \\
%\scalebox{0.9}{(a)}
%\end{tabular} 
%\hspace{10mm}
%\begin{tabular}{@{\hspace{5mm}}c@{}}
\includegraphics[width=0.35\textwidth]{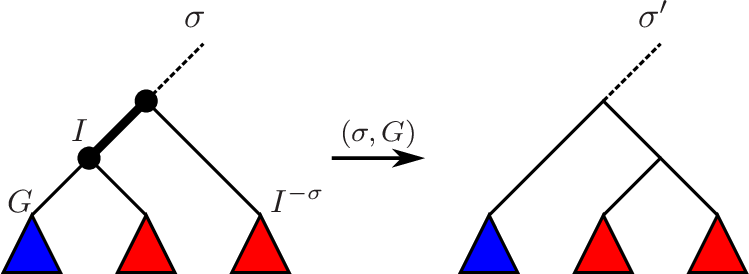} \\
% \scalebox{0.9}{(b)}
%\end{tabular}
%\\[12mm]

\vspace{1.5em}
\includegraphics[width=.85\textwidth]{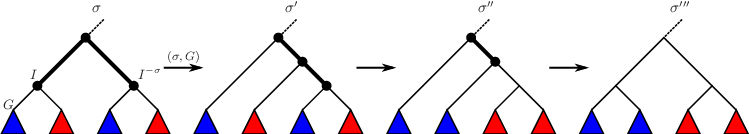} 
\\% \scalebox{0.9}{(c)}
%\end{tabular}
%\vspace{-2mm}
\end{center}
\vspace{-2mm}
\caption{Different types of incompatibility~ ---~ Type 1 (above) and  Type 2 (below)~ ---~ of a tree $\treeA\in\BT{\indexset}(K)$, $K=K_1\cup K_2$, with the split $\mathbb{K}=\{K_1,K_2\}$, and the NNI moves suggested by $\projector{\mathbb{K}}$ to resolve them. Clusters are colored blue or red according to their being contained in $K_1$ or $K_2$, respectively. Thickened vertices represent the recombinants affected by these moves.}
\label{fig:DeepXSplit}
\end{figure}
The following elementary property of edges in $\projector{\mathbb{K}}$ is crucial:
\begin{lemma}\label{lem:norm changes in projector graph} Suppose $\treeA\in\BT{\indexset}(K)$, $(\treeA,G)$ is an edge of $\projector{\mathbb{K}}$ and $\treeB=\NNI(\treeA,G)$. Then $\crossnorm{\mathbb{K}}{\treeA}=\crossnorm{\mathbb{K}}{\treeB}+1$.
\end{lemma}
\begin{proof} Let $I=\parentCL{G,\treeA}$ and let $J=\complementLCL{I}{\treeA}\cup\complementLCL{G}{\treeA}$ be the cluster replacing $I$ in $\treeB$. Also, set $M=\grandparentCL{G,\treeA}\in\cluster{\treeA}\cap\cluster{\treeB}$. 
In the transition from $\treeA$ to $\treeB$ only the clusters $I,J$ and $M$ change (or lose, or acquire) their child splits. Therefore:
\begin{equation}
	\crossnorm{\mathbb{K}}{\treeB}=
		\crossnorm{\mathbb{K}}{\treeA}
		-\xindex{\childCL{I,\treeA}}{\mathbb{K}}
		+\xindex{\childCL{J,\treeB}}{\mathbb{K}}
		-\xindex{\childCL{M,\treeA}}{\mathbb{K}}
		+\xindex{\childCL{M,\treeB}}{\mathbb{K}}
\end{equation}
\reffig{fig:DeepXSplit} demonstrates without loss of generality that, in the case when $I$ is of type 1 with respect to $\mathbb{K}$ the values of the above crossing indices are $0$, $0$, $1$ and $0$, respectively, resulting in a total decrease of one unit. The case when $I$ is of type $2$ produces the respective values of $0$, $1$, $3$ and $1$, also resulting in a total decrease of one unit.
\end{proof}
\begin{lemma}\label{lem:image of a projection} The following are equivalent for a vertex $\treeA\in\BT{\indexset}(K)$ of $\projector{\mathbb{K}}$:
\begin{enumerate}
	\item $\crossnorm{\mathbb{K}}{\treeA}>0$;
	\item $\projector{\mathbb{K}}$ contains an edge exiting $\treeA$;
	\item $\treeA\notin\BT{\indexset}(\mathbb{K})$.
\end{enumerate}
\end{lemma}
\begin{proof} First observe that, since $K$ is a cluster of $\treeA$, all clusters $I'\in\cluster{\treeA}$ not contained in $K$ have $\xindex{I'}{\mathbb{K}}=0$. 

\smallskip
\mbox{$(1)\THEN(2)$.\; } By the preceding observation, if $\crossnorm{\mathbb{K}}{\treeA}>0$ then $\treeA$ has a sub-cluster of $K$ whose child split is incompatible with $\mathbb{K}$. By \reflem{lem:recombinants_characterize_splits}, $\treeA$ then has a cluster $I\subsetneq K$ which is a recombinant of $\mathbb{K}$. Picking $G$ to be an appropriate $\treeA$-child of $I$ provides the required edge $(\treeA,G)$.

\smallskip
\mbox{$(2)\THEN(3)$.\; } Suppose $(\treeA,G)$ is an edge in $\projector{\mathbb{K}}$. Then $I:=\parentCL{G,\treeA}$ is incompatible with $\mathbb{K}$, proving (3).

\smallskip
\mbox{$(3)\THEN(1)$.\; } Finally, if $\treeA\notin\BT{\indexset}(\mathbb{K})$ then $\treeA$ contains a recombinant $I$ whose parent $M=\parentCL{I,\treeA}$ then must satisfy $\xindex{\childCL{M,\treeA}}{\mathbb{K}}>0$, resulting in $\crossnorm{\mathbb{K}}{\treeA}>0$.
\end{proof}

\begin{definition}[Projection] Let $\mathbb{K}=\{K_1,K_2\}$ be a partial split, and set $K=K_1\cup K_2$. For any $\treeA\in\BT{\indexset}(K)$ we define its \emph{projection to $\BT{\indexset}(K) \cap\BT{\indexset}(\mathbb{K})$} to be the tree $\treeC = \myproj{\treeA;\mathbb{K}}\in \BT{\indexset}(K) \cap\BT{\indexset}(\mathbb{K})$ whose clusters are of one of the following forms:
\begin{enumerate}[label=(\alph*)]
	\item $I\in\cluster{\treeA}$ with $I\cap K=\varnothing$ or $K\subseteq I$;
	\item $I\cap K_i\in\cluster{\treeC}$, $i\in\{1,2\}$ where $I\in\cluster{\treeA}$ (and $I\subseteq K$).
\end{enumerate}
\end{definition}
\begin{remark}\label{rem:projection is well defined} The tree $\myproj{\treeA;\mathbb{K}}$ is a well-defined binary tree in $\BT{\indexset}(K) \cap \BT{\indexset}(\mathbb{K})$ by \reflem{lem:BinTreeRestriction} (applied to $\BT{K}$).
\end{remark}

We are ready to state the main result of this section:
\begin{theorem}\label{thm:projectors work} The directed graph $\projector{\mathbb{K}}$ contains no directed cycles. Moreover, for every $\treeA\in\BT{\indexset}(K)$, every maximal directed path of $\projector{\mathbb{K}}$ emanating from $\treeA$ terminates at the tree $\myproj{\treeA;\mathbb{K}}\in \BT{\indexset}(K) \cap \BT{\indexset}(\mathbb{K})$ and has length $\crossnorm{\mathbb{K}}{\treeA}$.
\end{theorem}
\begin{proof} Denote $\Gamma\ldf\projector{\mathbb{K}}$ for short. By \reflem{lem:norm changes in projector graph}, the function $\crossnorm{\mathbb{K}}{\cdot}$ decreases by a unit along each edge of $\Gamma$, implying the absence of directed cycles in the graph. In particular, for each $\treeA\in\BT{\indexset}(K)$, the length of a directed path in $\Gamma$ emanating from $\treeA$ is bounded above by $\crossnorm{\mathbb{K}}{\treeA}$. Since, by \reflem{lem:image of a projection}, $\treeA\in\BT{\indexset}(K)$ has an exiting edge in $\Gamma$ if and only if $\crossnorm{\mathbb{K}}{\treeA}>0$, we conclude that all maximal directed paths in $\Gamma$ emanating from $\treeA$ have length exactly $\crossnorm{\mathbb{K}}{\treeA}$ and terminate in $\BT{\indexset}(K) \cap \BT{\indexset}(\mathbb{K})$.

It will be useful to henceforth denote
\begin{equation}\label{eqn:projector paths}
	\projpaths{\mathbb{K}}{\treeA}\ldf
	\left\{
		\mathbf{p}\,\big|\,
		\mathbf{p}\text{ is a maximal directed path in }\projector{\mathbb{K}}\text{ emanating from }\treeA
	\right\}
\end{equation}
It remains to prove that every $\mathbf{p}\in\projpaths{\mathbb{K}}{\treeA}$ terminates in $\myproj{\treeA;\mathbb{K}}$.

We will prove the remaining assertion of the proposition by induction on $\crossnorm{\mathbb{K}}{\treeA}$. More precisely, for any non-negative integer $k$ let $S(k)$ denote the statement that for every $\treeB\in\BT{\indexset}(K)$ satisfying $\crossnorm{\mathbb{K}}{\treeB}\leq k$ every path in $\projpaths{\mathbb{K}}{\treeB}$ terminates in $\myproj{\treeB;\mathbb{K}}$. Observing that $S(0)$ holds true by construction, we assume $S(k)$ holds for some $k\geq 0$ and deduce $S(k+1)$. 

Suppose $\treeA$ has $\crossnorm{\mathbb{K}}{\treeA}=k+1$. Once again, consider any directed edge $(\treeA,G)$ in $\projector{\mathbb{K}}$, and write $\treeB=\NNI(\treeA,G)$ with $\crossnorm{\mathbb{K}}{\treeB}=k$. Let $\treeC$ and $\treeC'$ denote the projections of $\treeA$ and $\treeB$ to $\BT{\indexset}(K) \cap \BT{\indexset}(\mathbb{K})$. Finally, letting $I=\parentCL{G,\treeA}$ and $J=\complementLCL{G}{\treeA}\cup\complementLCL{I}{\treeA}$ we recall that $\cluster{\treeB}=(\cluster{\treeA}\minus\{I\})\cup\{J\}$. We observe the following:
\begin{itemize}
	\item For any set $Q\subseteq\indexset$ satisfying $Q\cap K=\varnothing\;\vee\; K\subseteq Q$	and for any tree $\treeA'$ lying on a path in $\projpaths{\mathbb{K}}{\treeA}$~ ---~ for the trees $\treeB,\treeC$ and $\treeC'$ in particular~ ---~ one has $Q\in\cluster{\treeA}$ if and only if $Q\in\cluster{\treeA'}$. Thus, $\cluster{\treeC}\setminus\cluster{\treeC'}$ consists only of proper subsets of $K$.
	\item For a cluster $Q\subseteq K$ of $\treeA$ with $Q\neq I$ we have $Q\cap K_i\in\cluster{\treeC}\myTHEN Q\cap K_i\in\cluster{\treeC'}$ for $i\in\{1,2\}$ because $\cluster{\treeA}\minus\{I\}\subset\cluster{\treeB}$.
	\item Finally, we consider the clusters $I\cap K_i$: since $I\in\DeepXSplit{\treeA;\mathbb{K}}$, the sets $I\cap K_i$ are precisely the children of $I$ in $\treeA$, which makes them clusters of $\treeB$; since $I\cap K_i\subset K_i$, they are also clusters of $\treeC'$.	
\end{itemize}
To summarize, we have found out that $\cluster{\treeC}\subseteq\cluster{\treeC'}$. By the maximality of $\cluster{\treeC}$ as a nested family (\refrem{rem:MaximumCardinality} and \refrem{rem:projection is well defined}) they must be equal and we conclude that $\treeC=\treeC'$. Applying the induction hypothesis, we deduce that every path in $\projpaths{\mathbb{K}}{\treeA}$ starting with the edge $(\treeA,G)$ must terminate in $\treeC$. Since the choice of edge $(\treeA,G)$ was arbitrary, we are done.
\end{proof}

%%%%%%%%%%%%%%%%%%%%%%%%%%%%%%%%%%%%%%%%%%%
%%%%%%%%%%%%%%%%%%%%%%%%%%%%%%%%%%%%%%%%%%%
\subsection{The Navigation Distance}
\label{sec:dnav} 
%%%%%%%%%%%%%%%%%%%%%%%%%%%%%%%%%%%%%%%%%%%
%%%%%%%%%%%%%%%%%%%%%%%%%%%%%%%%%%%%%%%%%%%

The following result has the flavor of a commutation relation between {\it different} projector graphs:
\begin{lemma}\label{lem:nested projectors commute} Fix a pair of distinct partial splits $\mathbb{K}=\{K_1,K_2\}$ and $\mathbb{L}=\{L_1,L_2\}$. Setting $K=K_1\cup K_2$ and $L=L_1\cup L_2$ assume in addition that
$\{K,K_1,K_2\}\compatible\{L,L_1,L_2\}$. Then, for any $\treeA\in\BT{\indexset}(K)$ and any edge $(\treeA,G)\in\projector{\mathbb{K}}$ one has $\crossnorm{\mathbb{L}}{\NNI(\treeA,G)}=\crossnorm{\mathbb{L}}{\treeA}$.
\end{lemma}
\begin{proof} As before, set $\treeB=\NNI(\treeA,G)$ and consider the sets $I=\parentCL{G,\treeA}$, $J=\complementLCL{G}{\treeA}\cup\complementLCL{I}{\treeA}$ and $M=\grandparentCL{G,\treeA}$~ ---~ all contained in the cluster $K\in\cluster{\treeA}\cap\cluster{\treeB}$~ ---~ and recall that $\cluster{\treeB}=(\cluster{\treeA}\minus\{I\})\cup\{J\}$. Without loss of generality, $G\subseteq K_1$ and $\complementLCL{G}{\treeA}\subseteq K_2$.

Once again we observe that the transition from $\treeA$ to $\treeB$ affects only the crossing indices of the clusters $I,J,M$ (which are all contained in $K$) as follows:
\begin{equation}
	\crossnorm{\mathbb{L}}{\treeB}=
		\crossnorm{\mathbb{L}}{\treeA}
		-\underbrace{\xindex{\childCL{I,\treeA}}{\mathbb{L}}}_{\alpha}
		+\underbrace{\xindex{\childCL{J,\treeB}}{\mathbb{L}}}_{\beta}
		-\underbrace{\xindex{\childCL{M,\treeA}}{\mathbb{L}}}_{\gamma}
		+\underbrace{\xindex{\childCL{M,\treeB}}{\mathbb{L}}}_{\delta}
\end{equation}
Note that $K\neq L$, since otherwise the compatibility assumption and  \reflem{lem:SplitCompatibleEquivalence} would have forced $\mathbb{K}=\mathbb{L}$. 

Suppose now that $K\cap L=\varnothing$. In this case the restrictions of $\mathbb{L}$ to $I,J,M$ are all trivial and the corresponding crossing indices are all zero.

Suppose $K\subsetneq L$. Then, without loss of generality, we have $K\subseteq L_1$ by \reflem{lem:SplitCompatibleEquivalence} and all children of $I,J,M$ in $\treeA$ and $\treeB$ (as relevant) are compatible with $\mathbb{L}$, resulting again in zero crossing indices.

Since $K\compatible L$, $K\neq L$, we need only consider two cases (we refer the reader again to \reffig{fig:DeepXSplit} for an illustration):
\begin{itemize}
	\item\mbox{\bf $L\subseteq K_1$. } We have $\childCL{I,\treeA}\res{L}=\{G\cap L,\varnothing\}$ and therefore $\alpha=0$. Also, $\childCL{J,\treeB}\res{L}=\{\varnothing,\complementLCL{I}{\treeA}\cap L\}$, so that $\beta=0$. Finally, $\childCL{M,\treeA}\res{L}=\childCL{M,\treeB}\res{L}=\{G\cap L,\complementLCL{I}{\treeA}\cap L\}$ produces $\gamma=\delta$.
	\item\mbox{\bf $L\subseteq K_2$. } In this case we have $\childCL{I,\treeA}\res{L}=\{\varnothing,\complementLCL{G}{\treeA}\cap L\}$ and $\alpha$ is zero again. Similarly, observe that $\childCL{M,\treeB}\res{L}=\{\varnothing,J\cap L\}$ gives $\delta=0$. At the same time, $\childCL{J,\treeB}\res{L}=\childCL{M,\treeA}\res{L}=\{\complementLCL{G}{\treeA}\cap L,\complementLCL{I}{\treeA}\cap L\}$, so that $\beta=\gamma$.
\end{itemize}
This finishes the proof.
\end{proof}

Any pair of binary trees in $\BT{\indexset}$ has a common cluster (the cluster $\indexset$, for example), and one might hope to quantify the discrepancy between a pair of trees by counting common clusters which split differently in the two trees (perhaps, somehow accounting for the depth of these clusters). This motivates:
\begin{definition}\label{def:CrossingSplits}
For any $\treeA,\treeC \in \bintreetopspace_{J}$, let $\CommonXSplits{\treeA,\treeC}$ denote the set
\begin{equation}\label{eq:CrossingSplits} 
	\CommonXSplits{\treeA,\treeC} \ldf \crl{
		\CommonXCL \in \cluster{\treeA} \cap \cluster{\treeC} \big | \, \childCL{\CommonXCL,\treeA} \neq \childCL{\CommonXCL,\treeC}
		}.
\end{equation}
\end{definition}
\begin{remark}\label{rem:common crossing splits characterize equality} It is easy to see that, in $\BT{\indexset}$, $\treeA=\treeB$ if and only if $\CommonXSplits{\treeA,\treeB}=\varnothing$.
\end{remark}
\begin{corollary} For all $\treeA,\treeB\in\BT{\indexset}$ we have $\displaystyle \CommonXSplits{\treeA,\treeB} \ldf \crl{ \CommonXCL \in \cluster{\treeA} \cap \cluster{\treeB} \big | \, \childCL{\CommonXCL,\treeA} \not\compatible \childCL{\CommonXCL,\treeB} }$.
\end{corollary}
\begin{proof} Follows directly from \reflem{lem:SplitCompatibleEquivalence} and the definitions.
\end{proof}

Given a prescribed target tree $\treeB\in\BT{\indexset}$, the projector graphs introduced above give rise to a tool for achieving planned reductions in the number of clusters in $\CommonXSplits{\treeA,\treeB}$ at a given depth, for any tree $\treeA\in\BT{\indexset}$. More formally, consider the following construction:
\begin{definition}[Navigation Graph]\label{def:navigation to a point} Let $\treeB\in\BT{\indexset}$. Then $\NNIdir{\indexset,\treeB}$ denotes the directed subgraph of the NNI graph $\NNIdir{\indexset}$ with vertex set $\BT{\indexset}$ and all the edges $(\treeA,G)$ for which there exists a cluster $K\in\CommonXSplits{\treeA,\treeB}$ satisfying $(\treeA,G)\in\Gamma_\indexset(\childCL{K,\treeB})$.
\end{definition}
We proceed to prove statements about the navigation graph analogous to those we have shown to hold for the projector graphs. It is time to introduce:
\begin{definition}[Navigation Distance]\label{def:navigation distance} Let $\treeA,\treeB\in\BT{\indexset}$. We define the {\it navigation distance} from $\treeA$ to $\treeB$ to be:
\begin{eqnarray}
	\distNav\prl{\treeA,\treeB}
	&\ldf&	\sum_{K\in\cluster{\treeB}}
		\crossnorm{\childCL{K,\treeB}}{\treeA}\label{eqn:distNav crossnorms}\\
	&=&	\sum_{K\in\cluster{\treeB}}\sum_{L\in\cluster{\treeA}}
		\xindex{\childCL{L,\treeA}}{\childCL{K,\treeA}}\label{eqn:distNav xindex}
\end{eqnarray}
We also define the \emph{special crossing matrix} $\SMat\prl{\treeA,\treeB}$ by
\begin{equation}\label{eq.SMat}
\SMat\prl{\treeA,\treeB}_{K,L}:=\xindex{\childCL{L,\treeA}}{\childCL{K,\treeB}}, \quad \forall K \in \cluster{\treeA}, L \in \cluster{\treeB}.
\end{equation}
Thus, $\distNav$ coincides with the standard $1$-norm of the special crossing matrix.
\end{definition}

\begin{theorem}\label{thm:navigation graphs work} For any $\treeB\in\BT{\indexset}$ the graph $\NNIdir{\indexset,\treeB}$ has no directed cycles. Moreover, for any $\treeA\in\BT{\indexset}$ every maximal directed path in $\NNIdir{\indexset,\treeB}$ emanating from $\treeA$ terminates in $\treeB$ and has length $\distNav\prl{\treeA,\treeB}$. We will refer to such paths as \emph{navigation paths from $\treeA$ to $\treeB$}. 
\end{theorem}
\begin{proof} First, observe from equation \eqref{eqn:distNav crossnorms} that $\distNav\prl{\treeA,\treeB}$ is zero if and only if $\crossnorm{\mathbb{K}}{\treeA}=0$ for every pair $\mathbb{K}$ of siblings in $\treeB$. By \reflem{lem:image of a projection}, this is equivalent to saying that $\treeA\in\BT{\indexset}(\mathbb{K})$ for every pair of siblings in $\treeB$, or, in other words, that $\treeA=\treeB$. Moreover, note that $\distNav\prl{\treeA,\treeB}>0$ implies there is an edge of $\NNIdir{\indexset,\treeB}$ exiting $\treeA$: indeed, if $\treeA\neq\treeB$ then there exists a $K\in\CommonXSplits{\treeA,\treeB}$ (\refrem{rem:common crossing splits characterize equality}), so that $\treeA\notin\BT{\indexset}(\childCL{K,\treeB})$; \reflem{lem:image of a projection} guarantees an edge of $\projector{\childCL{K,\treeB}}$ exiting $\treeA$, which, by definition, is also an edge of $\NNIdir{\indexset,\treeB}$.

Suppose now $(\treeA,G)$ is an edge in $\NNIdir{\indexset,\treeB}$. That is, there exists $K\in\CommonXSplits{\treeA,\treeB}$ such that $(\treeA,G)\in\projector{\mathbb{K}}$ where $\mathbb{K}=\childCL{K,\treeB}$. 

Suppose there were more than one such $K$, that is: suppose $K,L\in\CommonXSplits{\treeA,\treeB}$, $K\neq L$, such that $I:=\parentCL{G,\treeA}$ is contained in both $K$ and $L$, and such that $\childCL{I,\treeA}$ is incompatible both with $\childCL{K,\treeB}$ and $\childCL{L,\treeB}$. Since $\varnothing\neq I\subseteq K\cap L$ and $K\compatible L$, we may assume $K\subsetneq L$. But then $K,L\in\cluster{\treeB}$ and $K\neq L$ implies $K$ is contained in a $\treeB$-child of $L$, denoted $L_1$. As $I\subseteq K$, we conclude that both $\treeA$-children of $I$ are contained in $L_1$~ ---~ a contradiction to the assumption that $\childCL{I,\treeA}$ and $\childCL{L,\treeB}$ are incompatible.

Let $\treeA'=\NNI(\treeA,G)$. Then, by \reflem{lem:norm changes in projector graph}, we have $\crossnorm{\childCL{K,\treeB}}{\treeA'}=\crossnorm{\childCL{K,\treeB}}{\treeA}-1$. Moreover, \reflem{lem:nested projectors commute} guarantees $\crossnorm{\childCL{L,\treeB}}{\treeA'}=\crossnorm{\childCL{L,\treeB}}{\treeA}$ for all $L\in\cluster{\treeB}$, $L\neq K$. Applying equation \eqref{eqn:distNav crossnorms} we obtain
\begin{equation}
	\distNav\prl{NNI(\treeA,G),\treeB}=\distNav\prl{\treeA,\treeB}-1\,.
\end{equation}
Thus, $\NNIdir{\indexset,\treeB}$ contains no directed cycles, and every maximal directed path in $\NNIdir{\indexset,\treeB}$ emanating from a fixed $\treeA\in\BT{\indexset}$ terminates after precisely $\distNav\prl{\treeA,\treeB}$ steps. By the preceding paragraph, every such path may only terminate in $\treeB$.
\end{proof}
The solution to the navigation problem implied by this theorem yields the following (very crude) bounds on the performance of the corresponding reactive navigation algorithm:
\begin{corollary}[Navigation Complexity] \label{cor:NNINavComplexity} The length of a navigation path through $\NNIgraph_{\indexset}$ does not exceed $\bigO{\card{\indexset}^2}$.
\end{corollary}
\begin{proof} Let $n=\card{\indexset}$. For all $\treeA,\treeB\in\BT{\indexset}$ we have $\card{\cluster{\treeB}}=\bigO{n}$, implying $\SMat\prl{\treeA,\treeB}$ has $\bigO{n^2}$ entries. The value of $\distNav\prl{\treeA,\treeB}$ never exceeds three times the number of entries in $\SMat\prl{\treeA,\treeB}$.
\end{proof}
\begin{corollary}\label{cor:NNINavMoveCost} Given $\treeA,\treeB\in\BT{\indexset}$, computing an edge of $\NNIdir{\indexset,\treeB}$ exiting $\treeA$ may be done in $\bigO{\card{\indexset}}$ time.
\end{corollary}
\begin{proof} Using a look-up table for the clusters of $\treeB$ \cite{day_joc1985}, a cluster $K\in\CommonXSplits{\treeA,\treeB}$ may be found in linear time by a traversal of $\treeA$. Next, an appropriate recombinant cluster may be found in linear time by post-order traversal of $\treeA\res{K}$ (compare with proof of \refprop{prop:CCdistanceCost}).
\end{proof}

The last theorem emphasizes the crucial role of the fact that all navigation paths from $\treeA$ to $\treeB$ have the same length, equal to $\distNav\prl{\treeA,\treeB}$, irrespective of the order in which one chooses to resolve the incompatibilities between the two trees. 
We will now consider additional applications of the last theorem which will help us clarify the geometry of the navigation distance and its relationship to the other dissimilarities mentioned in this paper.
\begin{lemma}\label{lem:dist to root split} Let $\mathbb{K}=\{K_1,K_2\}$ be a partial split, let $\treeB\in \BT{\indexset}\prl{K_1, K_2, K_1\cup K_2}$ and $\treeA\in\BT{\indexset}(K_1\cup K_2)$. Then:
\begin{enumerate}[label=(\alph*)] 
	\item \label{it.dist2rootsplit1} $\projector{\mathbb{K}}$ is contained in $\NNIdir{\indexset,\treeB}$;
	\item \label{it.dist2rootsplit2} Let $\treeA'=\myproj{\treeA;\mathbb{K}}$, then:
	\begin{equation}
		\distNav\prl{\treeA,\treeB}=
			\distNav\prl{\treeA,\treeA'}+
			\distNav\prl{\treeA',\treeB}
	\end{equation}
	\item \label{it.dist2rootsplit3} Finally, $\distNav\prl{\big.\treeA, \BT{\indexset}\prl{K_1, K_2, K_1\cup K_2}}=\crossnorm{\mathbb{K}}{\treeA}$.
\end{enumerate}
\end{lemma}
\begin{proof} For \reflem{lem:dist to root split}\ref{it.dist2rootsplit1}, let $(\treeA,G)$ be an edge of $\projector{\mathbb{K}}$. In particular, $\treeA\notin\BT{\indexset}(\mathbb{K})$ so that $K\in\CommonXSplits{\treeA,\treeB}$ which produces $(\treeA,G)\in\NNIdir{\indexset,\treeB}$ by definition. 

For \reflem{lem:dist to root split}\ref{it.dist2rootsplit2}, let $\mathbf{p}$ be a maximal path in $\projector{\mathbb{K}}$ emanating from $\treeA$. Then the endpoint of $\mathbf{p}$ is $\treeA':=\myproj{\treeA;\mathbb{K}}$ by \refthm{thm:projectors work}. Now apply \reflem{lem:dist to root split}\ref{it.dist2rootsplit1} and \refthm{thm:navigation graphs work} to extend $\mathbf{p}$ to a navigation path $\widetilde{\mathbf{p}}$ in $\NNIdir{\indexset,\treeB}$ from $\treeA$ to $\treeB$. Then:
\begin{equation}\label{eqn:projections can come first}
	\distNav\prl{\treeA,\treeB}=
	\length{\widetilde{\mathbf{p}}}=
	\length{\mathbf{p}}+\distNav\prl{\treeA',\treeB}=
	\treenorm{\mathbb{K}}{\treeA}+\distNav\prl{\treeA',\treeB}\,,
\end{equation}
as required.

Finally, for \reflem{lem:dist to root split}\ref{it.dist2rootsplit3}, pick $\treeB$ above to be a tree of $\BT{\indexset}\prl{K_1, K_2, K_1\cup K_2}$ with $\distNav\prl{\treeA,\treeB}$ minimal. By the construction above, $\treeA'\in \BT{\indexset}\prl{K_1, K_2,K_1\cup K_2}$ satisfies $\distNav\prl{\treeA,\treeA'}\leq\distNav\prl{\treeA,\treeB}$ while $\mathbf{p}$ is a navigation path from $\treeA$ to $\treeA'$. Thus $\treeA'$ must coincide with $\treeB$, and \eqref{eqn:projections can come first} reduces to the desired equality.
\end{proof}
\begin{corollary}\label{cor:NavSplitDistanceCost}
For any bipartition $\crl{L,R}$ of $\indexset$ and $\treeA \in \BT{\indexset}$, the navigation distance $\distNav\prl{\treeA, \BT{\indexset}\prl{L,R}}$ can be computed in linear time, $\bigO{\card{\indexset}}$. 
\end{corollary}
\begin{proof} Similarly to the proof of \refprop{prop:CMdistanceCost}, the crossing indices of $\treeA$-clusters with $\crl{L,R}$ can be determined in $\bigO{\card{\indexset}}$ time using  \reflem{lem:SplitCompatibleEquivalence} and by post order traversal of $\treeA$.
Therefore, by \reflem{lem:dist to root split} and \refthm{thm:projectors work}, the quantity $\distNav\prl{\treeA, \BT{\indexset}\prl{\indexset_L, \indexset_R}}$ can be computed in $\bigO{\card{\indexset}}$ by a complete traversal of $\treeA$.
\end{proof}

\begin{lemma} \label{lem.NavSplitDistCost}
For any bipartition  $\crl{L,R}$ of $\indexset$ and $\treeA \in \BT{\indexset}$,  
an NNI navigation path in $\projector{L,R}$ joining $\treeA$ to $ \BT{\indexset}\prl{L,R}$ can be computed in $\bigO{\card{\indexset}}$ time.
\end{lemma}
\begin{proof}
As illustrated in \reffig{fig:CrossingClusters}, since $\ancestorCL{I,\treeA} \subseteq \CLXSplitSet{\treeA; \crl{L,R}} \cup \crl{\indexset}$ for any $I \in \CLXSplitSet{\treeA; \crl{L,R}}$, the vertices and branches of $\treeA$ associated with clusters in $\CLXSplitSet{\treeA; \crl{L,R}} \cup \crl{\indexset}$ defines a tree structure, containing all the information required to compute the navigation distance $\distNav\prl{\treeA, \BT{\indexset}}\prl{L,R} = \norm{\treeA}_{\crl{L,R}}$ (\reflem{lem:dist to root split}.\ref{it.dist2rootsplit3}).
Hence, one can construct an NNI navigation path by a complete post-order traversal of this tree structure as follows:

\begin{enumerate}[noitemsep]
\item Set $k \leftarrow 0$ and $\treeA_0 \leftarrow \treeA$, and compute $\CLXSplitSet{\treeA_0;\crl{L,R}}$.
\item Find a cluster $I_0 \in \DeepXSplit{\treeA_0; \crl{L,R}}$ by a  post-order traversal of incompatible clusters $\CLXSplitSet{\treeA_0;\crl{L,R}}$ of $\treeA_0$.
\item While ($\CLXSplitSet{\treeA_k;\crl{L,R}} \neq \varnothing$)
\begin{enumerate}[noitemsep]
\item If $I_k \in \DeepXSplit{\treeA_k; \crl{L,R}}$ is Type 1, then, as illustrated in \reffig{fig:DeepXSplit}(top), choice a grandchild $G_k \in \childCL{I_k, \treeA_k}$ such that $\complementLCL{G_k}{\treeA_k}, \complementLCL{I_k}{\treeA_k} \subseteq L$ or $\complementLCL{G_k}{\treeA_k}, \complementLCL{I_k}{\treeA_k} \subseteq R$, and set
\begin{equation}
\treeA_{k+1} \leftarrow \NNI\prl{\treeA_k, G_k}, \quad
\CLXSplitSet{\treeA_{k+1};\crl{L,R}} \leftarrow \CLXSplitSet{\treeA_k;\crl{L,R}} \setminus\crl{I_k}, \quad
I_{k+1} \leftarrow \parentCL{I_k, \treeA_k}, \quad k \leftarrow k+1. \nonumber
\end{equation}
 
\item If $I_k \in \DeepXSplit{\treeA_k; \crl{L,R}}$ is Type 2, then, as illustrated in \reffig{fig:DeepXSplit}(bottom), choice $G_k \in \childCL{I_k, \treeA_k}$ and  $G_{k+1} \in \childCL{\complementLCL{I_k}{\treeA_k}, \treeA_k}$ such that $G_k, G_{k+1} \subseteq L$ or $G_k, G_{k+1} \subseteq R$, and $G_{k+2} = \complementLCL{G_k}{\treeA_k} \cup \complementLCL{G_{k+1}}{\treeA_k}$; and set
\begin{align}
\treeA_{k+1} \leftarrow \NNI\prl{\treeA_k, G_k}, \quad \treeA_{k+2} \leftarrow \NNI\prl{\treeA_{k+1}, G_{k+1}}, \quad \treeA_{k+3} \leftarrow \NNI\prl{\treeA_{k+2}, G_{k+2}}, \nonumber \\
 \CLXSplitSet{\treeA_{k+3};\crl{L,R}} \leftarrow \CLXSplitSet{\treeA_k;\crl{L,R}} \setminus\crl{I_k, \complementLCL{I_k}{\treeA_k}}, \quad  I_{k+3} \leftarrow \parentCL{I_k, \treeA_k}, \quad k \leftarrow k+3. \nonumber
\end{align}

\item Otherwise ($I_k$ and $\complementLCL{I_k}{\treeA_k}$ are Type 2 with $\childCL{I_k, \treeA_k} \compatible \crl{L,R}$ and $\childCL{\complementLCL{I_k}{\treeA_k}, \treeA_k} \not\compatible \crl{L,R}$), find a cluster $J_k \in \DeepXSplit{\treeA_k; \crl{L,R}}$ by a post-order traversal of incompatible clusters of the subtree of $\treeA_k$ rooted at $\complementLCL{I_k}{\treeA_k}$, and set $I_k \leftarrow J_k$.
\end{enumerate}
\item Return $\prl{\big.\treeA_k}_{k \in \brl{0, \,\norm{\treeA}_{\crl{L,R}}}}$ as an NNI navigation path starting at $\treeA$ and ending in $\BT{\indexset}\prl{L,R}$. 
\end{enumerate}

As discussed in the proof of \refprop{prop:CMdistanceCost}, all clusters of $\treeA$ incompatible with $\crl{L,R}$, i.e. $\CLXSplitSet{\treeA;\crl{L,R}}$ in Step 1, can be determined in $\bigO{\card{\indexset}}$ time.
Given $\CLXSplitSet{\treeA;\crl{L,R}}$, a cluster $I \in \DeepXSplit{\treeA; \crl{L,R}}$, in Step 2,  can be found in $\bigO{\card{\CLXSplitSet{\treeA; \crl{L,R}}}}\leq \bigO{\card{\indexset}}$ time  by a post-order traversal of incompatible clusters of $\treeA$.
Observe that the while loop terminates after at most $2\card{\CLXSplitSet{\treeA;\crl{L,R}}}$ iterations after a complete traversal of the tree structure defined by $\CLXSplitSet{\treeA;\crl{L,R}} \cup \crl{S}$  since $\card{\CLXSplitSet{\treeA_k;\crl{L,R}}}$ decreases at least by one unit after every two consecutive  iterations and 
a post-order subtree traversal in Step 3(c) is required only if the associated subtree is not explored yet.
Hence,  an NNI navigation path joining $\treeA$ to $ \BT{\indexset}\prl{L,R}$ can be found by a complete post-order traversal of $\treeA$ in $\bigO{\card{\indexset}}$ time.   
\end{proof}

The observation made in \reflem{lem:dist to root split} is a good example of how the dual representation of $\distNav$~ ---~ both in terms of paths in the NNI graph, and in terms of a closed-form formula quantifying inter-cluster incompatibility~ ---~ offers a practical compromise between the heretofore separate traditional approaches to constructing dissimilarities on $\BT{\indexset}$, those of edge comparison and of estimation of edit distances. A particular application of this dual nature is the decomposability of $\distNav$ (as defined in \cite{waterman_jtb1978}):
\begin{lemma}[Root Split Reduction]\label{lem:top split reduction} Fix $\treeB\in\BT{\indexset}$ and denote $\{L,R\}:=\childCL{\indexset,\treeB}$. Then for any $\treeA\in\BT{\indexset}$ one has:
\begin{equation}
	\distNav\prl{\treeA,\treeB}=
		\distNav\prl{\treeA,\BT{\indexset}(L,R)}+
		\distNav\prl{\treeA\res{L},\treeB\res{L}}+
		\distNav\prl{\treeA\res{R},\treeB\res{R}}
\end{equation}
\end{lemma}
\begin{proof} By \reflem{lem:dist to root split}(2) it suffices to prove 
\begin{equation}
	\distNav\prl{\myproj{\treeA;L,R},\treeB}=
		\distNav\prl{\treeA\res{L},\treeB\res{L}}+
		\distNav\prl{\treeA\res{R},\treeB\res{R}}
\end{equation}
By definition, $\cluster{\myproj{\treeA;L,R}}=\{\indexset\}\cup\cluster{\treeA\res{L}}\cup\cluster{\treeA\res{R}}$ so it suffices to prove: 
\begin{equation}
	\treeA\in\BT{\indexset}(L,R)\THEN
	\distNav\prl{\treeA,\treeB}=
		\distNav\prl{\treeA\res{L},\treeB\res{L}}+
		\distNav\prl{\treeA\res{R},\treeB\res{R}}
\end{equation}
At this stage, however, observe that $\cluster{\treeA\res{L}}$ and $\cluster{\treeA\res{R}}$ together exhaust the list of of clusters of $\treeA$ not equal to $\indexset$, with the same holding {\it ab initio} for $\treeB$. This allows us to finish the proof by applying \refthm{thm:navigation graphs work} separately in $\BT{L}$ and $\BT{R}$.
\end{proof}

The root split reduction of the NNI navigation dissimilarity may be used for its efficient computation:
\begin{corollary}\label{cor:distNavCost} The NNI navigation dissimilarity $\distNav$ on $\BT{\indexset}$ is computable in $\bigO{\card{\indexset}^2}$ time.
\end{corollary}
\begin{proof} Let $\treeA,\treeB\in\BT{\indexset}$ and $\crl{L,R} = \childCL{\indexset,\treeB}$. By the root split reduction above and the last corollary, $\distNav\prl{\treeA,\treeB}$ requires the computation of $\distNav\prl{\treeA, \BT{\indexset}\prl{\childCL{\indexset, \treeB}}}$ at a cost of $\bigO{\card{\indexset}}$ time, plus the computation of the restrictions $\treeA\res{L}$ and $\treeA\res{R}$, each of which can be computed using post-order traversal of $\treeA$ in $\bigO{\card{\indexset}}$ time.
Hence, computing $\distNav\prl{\treeA,\treeB}$ requires a complete (depth-first) traversal of $\treeB$ with each stage incurring at most a linear time cost in $\card{\indexset}$.
\end{proof}

\begin{corollary}
An NNI navigation path  joining $\treeA \in \bintreetopspace_{\indexset}$ to $\treeB \in \bintreetopspace_{\indexset}$ can be computed in $\bigO{\card{\indexset}^2}$ time.
\end{corollary}
\begin{proof}
Similar to the recursive expression of $\distNav$ in \reflem{lem:top split reduction}, an NNI navigation path joining $\treeA$ to $\treeB$ can 
be found using the decomposability property within a divide-and-conquer approach as follows: first obtain an NNI navigation path from $\treeA$ to $\bintreetopspace_{\indexset}\prl{\childCL{\indexset,\treeB}}$ in $\bigO{\card{\indexset}}$ (\reflem{lem.NavSplitDistCost}) and then find NNI navigation paths between subtrees.
Hence, this requires the pre-order traversal of $\treeB$ each of whose step costs $\bigO{\card{\indexset}}$. 
Thus, an NNI navigation path joining $\treeA$ to $\treeB$ can be recursively computed in $\bigO{\card{\indexset}^2}$ time, which completes the proof. 
\end{proof}

%%%%%%%%%%%%%%%%%%%%%%%%%%%%%%%%%%%%%%%%%%%%%%%%%%%%%%%%%
%%%%%%%%%%%%%%%%%%%%%%%%%%%%%%%%%%%%%%%%%%%%%%%%%%%%%%%%%
\subsection{Properties of the Navigation Dissimilarity}
%%%%%%%%%%%%%%%%%%%%%%%%%%%%%%%%%%%%%%%%%%%%%%%%%%%%%%%%%
%%%%%%%%%%%%%%%%%%%%%%%%%%%%%%%%%%%%%%%%%%%%%%%%%%%%%%%%%

%We now continue with a list of significant properties of $\distNav$ and its relation with other tree measures:  
%Now, one might wonder if $\distNav$ \refeqn{eq:NavPathLength}  defines a dissimilarity measure in $\bintreetopspace_{\indexset}$, and may be a metric. 
%
\begin{proposition} \label{rem:NavDistanceNotMetric}
The NNI navigation dissimilarity $\distNav$ is positive definite and symmetric, but it is not a metric. 
\end{proposition}
\begin{proof} That $\distNav$ is positive definite follows directly from its definition. \reflem{lem:xindex is symmetric} proves it is symmetric and \refcor{cor:NotMetric} with \reffig{fig:CMatNormNavLengthNotMetric} shows where the triangle inequality fails.
\end{proof}

%STOPPED HERE
\begin{lemma}\label{lem:UpperBoundSplitNavPath} Let $\crl{L,R}$ be a bipartition of $\indexset$ and $\treeA\in\BT{\indexset}$. Then we have the tight bound:
\begin{equation}\label{eq:UpperBoundSplitNavPath}
	\distNav\prl{\treeA,\BT{\indexset}\prl{L,R}}\leq
		\card{\indexset}+\min\prl{\card{L}, \card{R}}-3.  
\end{equation}

\end{lemma}
\begin{proof} Denote $\mathbb{S}=\{L,R\}$. For any $\treeA\in\BT{\indexset}$ and $I\in\cluster{\treeA}$ observe that (i) $\xindex{\childCL{I,\treeA}}{\mathbb{S}}=0$ if $I$ is a singleton or $\card{I}=2$, and (ii) otherwise for larger clusters $\xindex{\childCL{I,\treeA}}{\mathbb{S}}$ equals $3$ or $1$ only if, respectively,  both clusters or only one cluster of $\childCL{I,\treeA}$ are incompatible with $\mathbb{S}$. Since there are at least $\card{S}+1$ clusters of the first kind, there are at most $\card{S}-2$ clusters of the second kind. Thus, applying \reflem{lem:dist to root split} and \refthm{thm:projectors work} we have
\begin{equation}
	\distNav\prl{\treeA,\BT{\indexset}\prl{L,R}}\leq(\card{\indexset}-2)+\card{\mathcal{X}}\,,
\end{equation}
where $\mathcal{X}$ is the set of all $I\in\cluster{\treeA}$ both of whose children are incompatible with $\mathbb{S}$. For each $I\in\mathcal{X}$ both $I\cap L$ and $I\cap R$ are non-singleton clusters of $\treeA\res{L}$ and $\treeA\res{R}$, respectively (each child of $I$ intersects each of $L,R$). Suppose now that $I,J\in\mathcal{X}$ are distinct. There are two cases, without loss of generality:
\begin{itemize}
	\item If $I\cap J=\varnothing$, then $I\cap L\neq J\cap L$ (and similarly for $R$);
	\item If $I\subsetneq J$, then $J$ has a child $I'$ disjoint from $I$, and this child must intersect $L$. Hence, $I\cap L\subsetneq J\cap L$.
\end{itemize}
We conclude that the map $I\mapsto I\cap L$ (respectively $I\cap R$) of $\mathcal{X}$ to $\cluster{\treeA\res{L}}$ (resp. to $\cluster{\treeA\res{R}}$) is injective, and has no singleton clusters in its image. Thus, $\card{\mathcal{X}}\leq\min\prl{\card{L}-1, \card{R}-1}$, proving the desired inequality.
 
The example $\treeA, \treeB \in \BT{\brl{n}}$ 
in \reffig{fig.distCMdistNavDiameter} with $\crl{L,R}=\childCL{\brl{n}, \treeB} =\crl{\crl{1},\crl{2,3, \ldots, n}}$ shows that the upper bound in \eqref{eq:UpperBoundSplitNavPath} is tight (where $\distNav\prl{\treeA,\BT{\indexset}{\prl{L,R}}}=n-2$).
\end{proof}

\begin{proposition}\label{prop:NavDistDiameter} $\diam{\BT{\indexset}, \distNav} = \frac{1}{2} \prl{\card{\indexset} - 1} \prl{ \card{\indexset}  -2}\,$ . 
\end{proposition}
\begin{proof}
We proceed by induction over $\card{\indexset}$, with the base case $\card{\indexset} = 2$ satisfying $\card{\BT{\indexset}} = 1$. The formula then holds trivially, as $\distNav=0$.

For the induction step assume $\card{\indexset}\geq 3$ and that $\treeA,\treeB \in \BT{\indexset}$ satisfy $\distNav\prl{\treeressym{K}{\treeA},\treeressym{K}{\treeB}} \leq \frac{1}{2}\prl{ \card{K} - 1} \prl{ \card{K} - 2}$ for every $K\in \childCL{\indexset,\treeB}=\{L,R\}$. 

Let $\mu=\min\prl{\card{L},\card{R}}$, and note that $\card{L}\card{R}=\mu(\card{\indexset}-\mu)$. We now apply the root split reduction (\reflem{lem:top split reduction}):

\noindent
\begin{align}
\distNav\prl{\treeA,\treeB}  
&= \underbrace{\distNav\prl{\treeA, \BT{\indexset}\prl{L,R}}}_{ \substack{\text{by \reflem{lem:UpperBoundSplitNavPath}} \\ \leq \, \card{\indexset}+\mu-3} } 
+ \underbrace{ \distNav\prl{\treeressym{L}{\treeA},\treeressym{L}{\treeB}}}_{\substack{\text{by induction} \\ \leq \,\frac{1}{2}\,\prl{ \card{L} - 1}\, \prl{ \card{L} -2}}}  
+ \underbrace{\distNav\prl{\treeressym{R}{\treeA}, \treeressym{R}{\treeB}}}_{\substack{\text{by induction}\\ \leq \, \frac{1}{2}\, \prl{ \card{R} - 1} \,\prl{  \card{R} -2} }}, \\ 
%
%& \leq \tfrac{1}{2}\prl{\card{L}^2 + \card{R}^2} - \tfrac{3}{2}\underbrace{\prl{\card{L} + \card{R}}}_{\card{\indexset}} + \card{\indexset} +\mu - 1, \\
%
%& = \tfrac{1}{2}\bigg(	\card{\indexset}^2 -  2\card{L}\card{R}-3\card{\indexset}+2 \bigg) + \card{\indexset} + \mu - 2,\\
%
%& = \tfrac{1}{2}\bigg(	\card{\indexset}^2-3\card{\indexset}+2 \bigg) -\mu(\card{\indexset}-\mu)+ \card{\indexset} + \mu - 2,\\
%
&\leq \tfrac{1}{2} \prl{\card{\indexset}-1} \prl{ \card{\indexset} - 2}  + \underbrace{\prl{1 - \mu} \prl{\card{\indexset} - \mu-2}}_{\text{non-positive whenever } \card{\indexset} \geq 3},\\
&\leq \tfrac{1}{2} \prl{\card{\indexset}-1} \prl{ \card{\indexset} - 2}.
\end{align}

Finally, note that the trees in \reffig{fig.distCMdistNavDiameter} realize this bound on the diameter. 
\end{proof}

%%%%%%%%%%%%%%%%%%%%%%%%%%%%%%%%%%%%%%%%%%%%%%%%%%%%%%%%%%
%%%%%%%%%%%%%%%%%%%%%%%%%%%%%%%%%%%%%%%%%%%%%%%%%%%%%%%%%%
\subsection{Relations with Other Tree Measures}
%%%%%%%%%%%%%%%%%%%%%%%%%%%%%%%%%%%%%%%%%%%%%%%%%%%%%%%%%%
%%%%%%%%%%%%%%%%%%%%%%%%%%%%%%%%%%%%%%%%%%%%%%%%%%%%%%%%%%

Like $\distCM$ (\refprop{prop:CM_RF_bound}), $\distNav$ is tightly bounded in terms of $\distRF$ as follows:
\begin{proposition} \label{prop:Nav_RF_bound}
%For any $\treeA,\treeB \in \bintreetopspace_{\indexset}$,
Over $\bintreetopspace_{\indexset}$ one has $\distRF \leq \distNav \leq \tfrac{1}{2}\distRF^2 + \tfrac{1}{2}\distRF\,$ and both bounds are tight.
\end{proposition}
\begin{proof} Since $\distNav$ is realized by paths in the NNI graph we have $\distNNI\leq\distNav$. The lower bound then follows from $\distRF\leq\distNNI$ (\refcor{cor:NNI_RF_bound}). The bound is tight because
\begin{equation}
\distRF\prl{\treeA,\treeB} = 1 
\Leftrightarrow 
\distNNI\prl{\treeA,\treeB} = 1 
\Leftrightarrow 
\distNav\prl{\treeA,\treeB} = 1 . 
\end{equation}  
For the upper bound we argue by induction over $\card{\indexset}$, keeping in mind that for $\card{\indexset} = 2$ the result holds trivially. Suppose $\card{\indexset} \geq 3$. Now, if $\treeA$ and $\treeB$ have no common nontrivial clusters then $\distRF\prl{\treeA,\treeB} = \card{\indexset} - 2$ and the result follows from  \refprop{prop:NavDistDiameter}. Otherwise, let $I \in \cluster{\treeA} \cap \cluster{\treeB}$ be a nontrivial cluster and consider the tree $\treeA'$ obtained from $\treeA$ by replacing the branch $\treeA\res{I}$ with the branch $\treeB\res{I}$.

By theorem \refthm{thm:navigation graphs work} and by the definition of $\distRF$, respectively, we have:
\begin{eqnarray}
	\distNav\prl{\treeA,\treeB}&=&\distNav\prl{\treeA,\treeA'}+\distNav\prl{\treeA',\treeB}\\
	\distRF\prl{\treeA,\treeB}&=&\distRF\prl{\treeA,\treeA'}+\distRF\prl{\treeA',\treeB}
\end{eqnarray}
Let $\alpha=\distRF\prl{\treeA,\treeA'}$ and $\beta=\distRF\prl{\treeA',\treeB}$. Since $\distNav\prl{\treeA,\treeA'}=\distNav\prl{\treeA\res{I},\treeA'\res{I}}$ we may apply the induction hypothesis in $\BT{I}$ to conclude $\distNav\prl{\treeA,\treeA'}\leq\tfrac{1}{2}\alpha(\alpha+1)$. By pruning the trees $\treeA'$ and $\treeB$ at cluster $I$ we may apply the induction hypothesis in $\BT{\bar{S}}$, where $\bar{S}$ is the result of contracting $I$ to a single vertex, to conclude that $\distNav\prl{\treeA',\treeB}\leq\tfrac{1}{2}\beta(\beta+1)$. It then follows that:
\begin{equation}
\distNav\prl{\treeA,\treeB} \leq \tfrac{1}{2} \alpha \prl{\alpha +1} + \tfrac{1}{2} \beta \prl{\beta +1} 
\leq \tfrac{1}{2}\prl{\alpha + \beta}\prl{\alpha + \beta + 1} = \tfrac{1}{2}\distRF\prl{\treeA,\treeB} \prl{\distRF\prl{\treeA,\treeB} + 1}.
\end{equation}
\refprop{prop:NavDistDiameter} ensures this bound is tight. 
\end{proof}

\begin{proposition} \label{prop:Nav_CM_bound}
Over $\bintreetopspace_{\indexset}$ one has $\distNav\prl{\treeA,\treeB} \leq \frac{3}{2}\distCM\prl{\treeA,\treeB}\,$.
\end{proposition}
\begin{proof} Consider the closed form expression of $\distNav$ \refeqn{eqn:distNav xindex} in terms of crossing indices. Since the trivial clusters are compatible with any subset of $\indexset$, it will suffices to verify that, for each $I\in\cluster{\treeA}$ and $J\in\cluster{\treeB}$, one has: 
\begin{equation}
	\xindex{\childCL{I,\treeA}}{\childCL{J,\treeB}}\leq
	\tfrac{3}{2}\sum_{A\in\childCL{I,\treeA}}\sum_{B\in\childCL{J,\treeB}}\indicator\prl{A\not\compatible B}
\end{equation}
This verification is straightforward.
\end{proof}
The overall ordering of tree dissimilarities in \refcor{cor:NNI_RF_bound}, \refprop{prop:CM_CC_bound} and \refprop{prop:Nav_CM_bound} can be combined as:
\begin{theorem} \label{thm:DissimilarityOrder}
For non-degenerate hierarchies,
\begin{equation}
\frac{2}{3} \distRF \leq\frac{2}{3} \distNNI \leq \frac{2}{3}\distNav \leq \distCM \leq \distCC.
\end{equation}
\end{theorem}

Finally, we remark that the NNI navigation dissimilarity $\distNav$ (\refdef{def:navigation distance}) can be generalized to a pair of trees, $\treeA$ and $\treeB$, in $\treetopspace_{\indexset}$ as 
\begin{equation}\label{eq:NavDistanceAll}
\distNav\prl{\treeA, \treeB} = \frac{1}{2}\prl{\big.\norm{\big.\SMat\prl{\treeA,\treeB}}_1 + \norm{\big.\SMat\prl{\treeB,\treeA}}_1},
\end{equation} 
which is non-negative and symmetric.
For non-degenerate trees $\treeA,\treeB \in \bintreetopspace_{\indexset}$ one has $\SMat\prl{\treeA,\treeB} =  \tr{\SMat\prl{\treeB,\treeA}}$ (which is evident from \refeqn{eq.SMat} and \reflem{lem:xindex is symmetric}), so that $\distNav$ in \refeqn{eq:NavDistanceAll} simplifies back to \refeqn{eqn:distNav xindex}.\footnote{$\tr{\mat{A}}$ is the transpose of matrix $\mat{A}$.}
Although the closed form expression of $\distNav$ in \refthm{thm:navigation graphs work} enables the generalization of $\distNav$ to degenerate trees as above, the notion of NNI moves (\refdef{def:NNIMove}) is generally not valid in $\treetopspace_{\indexset}$.     

As for non-degenerate trees in \refprop{prop:Nav_CM_bound}, the generalized $\distNav$ in $\treetopspace_{\indexset}$ can be bounded above by $\distCM$ as follows: 

\begin{proposition}\label{prop:Nav_CM_bound_Degenerate}
Over $\treetopspace_{\indexset}$ one has $\distNav \leq \prl{\frac{1}{8}\card{\indexset}^2 + \frac{1}{4}\card{\indexset}} \distCM\,$.
\end{proposition}
\begin{proof}
Note that the number of nontrivial children of a cluster in a tree can be at most $\frac{1}{2}\card{\indexset}$. 
Hence one can verify the result following  similar steps as in the  proof of \refprop{prop:Nav_CM_bound}. 
\end{proof}

%%%%%%%%%%%%%%%%%%%%%%%%%%%%%%%%%%%%%%%
%%%%%%%%%%%%%%%%%%%%%%%%%%%%%%%%%%%%%%%
\section{Discussion and Statistical Analysis}
\label{sec:Discussion}
%%%%%%%%%%%%%%%%%%%%%%%%%%%%%%%%%%%%%%%
%%%%%%%%%%%%%%%%%%%%%%%%%%%%%%%%%%%%%%%

%%%%%%%%%%%%%%%%%%%%%%%%%%%%%%%%%%%%%%%
%%%%%%%%%%%%%%%%%%%%%%%%%%%%%%%%%%%%%%%
\subsection{Consensus Models and Median Trees}
%%%%%%%%%%%%%%%%%%%%%%%%%%%%%%%%%%%%%%%
%%%%%%%%%%%%%%%%%%%%%%%%%%%%%%%%%%%%%%%

Let us recall a definition : a \emph{median tree} of a set of sample trees is a tree whose sum of distances to the sample trees is minimum.
Although the notion of a median tree is simple and well-defined,  finding a median tree of a set of trees is generally a hard combinatorial problem. 
On the other hand, a consensus model of a set of sample trees is a computationally efficient tool to identify  common structures of sample trees.  
In particular, a  remark relating $\distCM$ and $\distNav$ to commonly used consensus models of a set of trees  and their median tree(s) is:

\begin{proposition}
Both the strict and loose consensus trees, $T_{\ast}$ and $T^{\ast}$, of any set of trees $T$  in $\treetopspace_{\indexset}$ (\refdef{def.StrictLooseConsensus}) are median trees with respect to both the crossing ($\distCM$) and navigation ($\distNav$) dissimilarities. In fact, for any $d \in \crl{\distCM, \distNav}$ one has:
\begin{equation}
\sum_{\tree \in T} \dist\prl{\tree, T_{\ast}} = \sum_{\tree \in T} \dist\prl{\tree, T^{\ast}} = 0.
\end{equation}
%
%Moreover, the loose consensus tree is the maximal (finest) median tree sharing each of its clusters with at least one sample tree.  
\end{proposition} 
\begin{proof}
By \refdef{def.StrictLooseConsensus}, both strict and loose consensus trees only contain clusters that are compatible with the clusters of every tree in $T$, and  the loose consensus tree is the finest median tree containing only clusters from the sample trees. Thus, the result  follows for both $\distCM$ and $\distNav$ due their relation in \refprop{prop:Nav_CM_bound_Degenerate}. 
\end{proof}

%%%%%%%%%%%%%%%%%%%%%%%%%%%%%%%%%%%%%%%%%%%%%%%%%%%%%%%
%%%%%%%%%%%%%%%%%%%%%%%%%%%%%%%%%%%%%%%%%%%%%%%%%%%%%%%
\subsection{Sample Distribution of Dissimilarities}
%%%%%%%%%%%%%%%%%%%%%%%%%%%%%%%%%%%%%%%%%%%%%%%%%%%%%%%
%%%%%%%%%%%%%%%%%%%%%%%%%%%%%%%%%%%%%%%%%%%%%%%%%%%%%%%

%As stated in \cite{lin_rajan_moret_TCBB2012}, there is no available biologically motivated benchmark dataset for the comparison of different tree measures. 
%We therefore resort to a standard statistical analysis of their distributions.

To compare their discriminative power, we use a standard statistical analysis of empirical distributions of different tree measures.
The shape of the distribution of a tree measure tells how informative it is; for example, a highly concentrated distribution means that the associated tree measure behaves like the discrete metric\footnote{The discrete metric $\dist:X \times X \rightarrow \R_{\geq 0}$ on a set $X$ is defined as for any $x \neq y \in X$ $\dist\prl{x,x} = 0$ and $\dist\prl{x,y} = 1$. } as in the case of the Robinson-Foulds distance  --- see \reffig{fig:TreeMeasureDistributionN25}. 
Finding a closed form expression for the distribution of a tree measure is a hard problem, and so extensive numerical simulations are generally applied to obtain its sample distribution.   
In particular, using the uniform and Yule model \cite{semple_steel_2003} for generating random trees, we compute the empirical distributions of $\distRF$, $\distMS$, $\distCC$, $\distCM$, and $\distNav$ as illustrated in \reffig{fig:TreeMeasureDistributionN25}.\footnote{In our numerical simulations for any chosen tree measure we observe the same pattern of sample distribution for different numbers of leaves, and so here we only include results for $\bintreetopspace_{\brl{25}}$.}   
Moreover, in \reftab{tab:SkewnessKurtosis} we present two commonly used statistical  measures, skewness and kurtosis, for describing the shapes of the probability distributions of all these tree measures. 
Here, recall that the skewness of a probability distribution  measures its tendency on one side of the mean, and the concept of kurtosis measures the peakedness of the distribution \cite{rice_2007}.
In addition to their computational advantage over $\distMS$, as illustrated in both \reffig{fig:TreeMeasureDistributionN25} and \reftab{tab:SkewnessKurtosis}, like $\distMS$, our tree measures, $\distCC$, $\distCM$ and $\distNav$, are significantly more discriminative, with wider ranges of values and symmetry, than $\distRF$. 
  
\begin{table}[hb]
\caption{Skewness and Kurtosis Values for the Distributions of Tree Measures in $\bintreetopspace_{\brl{25}}$}
\label{tab:SkewnessKurtosis}
\centering
\vspace{2mm}
\begin{tabular}{|l|c|c|c|c|}
\hline
 & \multicolumn{2}{c|}{Skewness} & \multicolumn{2}{c|}{ Kurtosis} \\ \hline
 & Uniform & Yule &Uniform  & Yule\\ \hline
$\distRF$ \refeqn{eq.RFdistance}&  $-2.6162$ & $-2.0740$ & $9.8609$ & $7.3998$ \\ \hline
$\distMS$ (Def. \ref{def.RF and MS distances}) & $0.1293$ & $-0.0117$  & $3.0060$ & $3.1136$\\ \hline
$\distCC$ \refeqn{eq.CCdistance} & $-0.9294$ & $-1.2507$  & $3.8601$ & $5.2724$\\ \hline
$\distCM$ (Def. \ref{def.CompatibilityAndCrossingMat}) & $0.1390$ & $-0.0405$  & $3.1275$ & $3.2103$ \\\hline 
$\distNav$ (Def. \ref{def:navigation distance}) & $0.8809$ & $-0.1195$  & $4.8707$ & $3.0746$\\ \hline
\end{tabular}
\end{table}

\begin{figure}[htb]
\centering
\begin{tabular}{c}
\includegraphics[width= 0.95\textwidth]{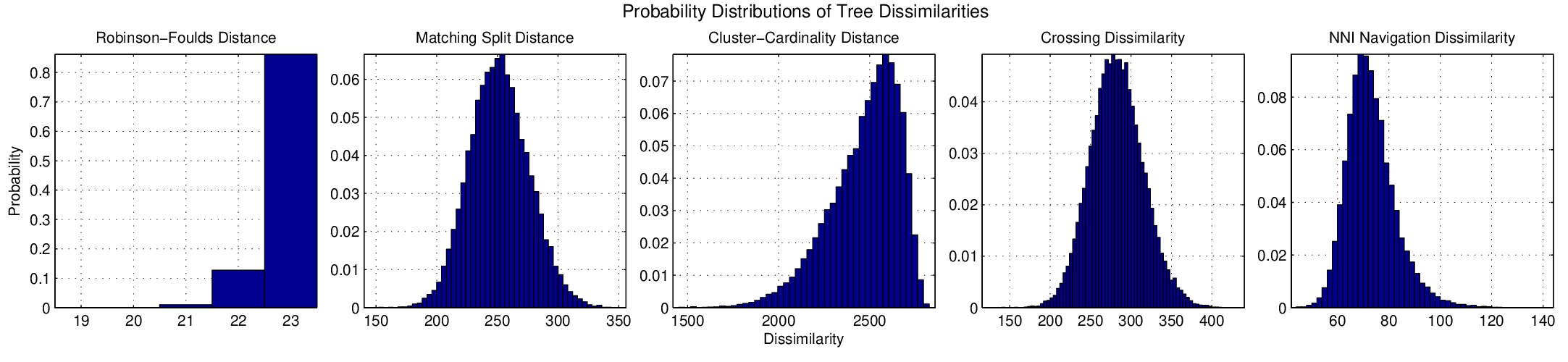} 
\\[-2.5mm]
%\includegraphics[width= 0.85\textwidth]{cdfTreeDissimilarityUniformN25.eps}
%\\
\scalebox{0.8}{(a)}
\\
\includegraphics[width= 0.95\textwidth]{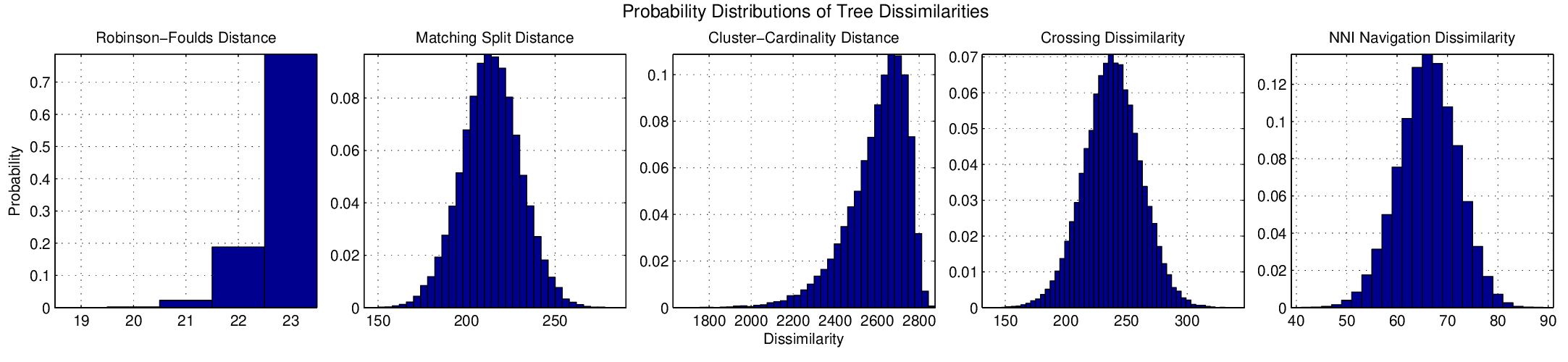} 
\\[-2.5mm]
%\includegraphics[width= 0.85\textwidth]{cdfTreeDissimilarityYuleN25.eps} 
%\\
\scalebox{0.8}{(b)}
\end{tabular}
\vspace{-3mm}
\caption{Empirical distribution of tree dissimilarities in $\bintreetopspace_{\brl{25}}$: (from left to right) the Robinson-Foulds distance $\distRF$ \refeqn{eq.RFdistance}, the matching split distance $\distMS$ (Def. \ref{def.RF and MS distances}), the cluster-cardinality distance $\distCC$ \refeqn{eq.CCdistance},  the crossing dissimilarity $\distCM$ (Def. \ref{def.CompatibilityAndCrossingMat}), and the NNI navigation dissimilarity $\distNav$ (Def. \ref{def:navigation distance}). 
100000 sample hierarchies  are generated using (a) the uniform  and (b)Yule model  \cite{semple_steel_2003}.
The resolutions of histograms of tree measures, from left to right, are 1, 4, 32, 4, 2 unit(s), respectively.
 }
\label{fig:TreeMeasureDistributionN25}
\end{figure}

%%%%%%%%%%%%%%%%%%%%%%%%%%%%%%
%%%%%%%%%%%%%%%%%%%%%%%%%%%%%%
\section{Conclusion}
\label{sec.Conclusion}
%%%%%%%%%%%%%%%%%%%%%%%%%%%%%%
%%%%%%%%%%%%%%%%%%%%%%%%%%%%%%

%This material below would go better in the conclusion since it is giving opinions about future work.  In any case, I don't understand what you've said here, so you'll have to re work it with me.
% 
%Actually, an NNI navigation path is likely to have an eventual significance in such a setting since
%the measures of dissimilarity of $\distRF$, $\distCM$ and  $\distCC$
%respecting  a desired
%hierarchy $\tau^*$,   are all non-increasing (in fact, strictly decreasing,
%for a suitably weighted version of $\distCM$) along any of the equal length NNI navigation trajectories that define $\distNav$.

This paper presents three new tree measures  for efficient discriminative comparison of trees.
First, using the well known relation between trees and ultrametrics, the cluster-cardinality metric $\distCC$ is constructed  as the pullback of matrix norms along an embedding of trees into the space of matrices.
Second, we present the crossing dissimilarity $\distCM$ that counts the pairwise incompatibilities of trees.
Third, the NNI navigation dissimilarity $\distNav$ while presented in closed form is constructed as the length of a navigation path in the space of trees.

All of our dissimilarities can be computed in $\bigO{n^2}$ with the number of leaves $n$, and they generalize to degenerate trees as well.
Moreover, we provide a closed form expression for each proposed dissimilarity and present an ordering relation between these tree dissimilarities and related tree metrics in the literature (\refthm{thm:DissimilarityOrder}).
%\begin{align}
%\frac{2}{3}\distRF \leq \frac{2}{3}\distNNI \leq \frac{2}{3}\distNav \leq \distCM \leq \distCM.
%\end{align}
Our numerical studies, summarized in \reffig{fig:TreeMeasureDistributionN25}, suggest that the proposed tree measures are significantly more informative and discriminative than the Robinson-Foulds distance $\distRF$, while maintaining a computational advantage over other distances such as the matching-split distance \cite{bogdanowicz_giaro_TCBB2012,lin_rajan_moret_TCBB2012}.   

Finally, the system of projector graphs (\refthm{thm:projectors work}) and navigation graphs (\refthm{thm:navigation graphs work}) seems to play a fundamental role in the geometry of the NNI graph, realizing many of the intuitive desiderata of tree dissimilarity measures that have accumulated in the literature over the years. Consequently, NNI navigation paths are likely of some significance for consensus/average models or statistical analysis of trees. 

%Finally, NNI navigation paths joining a pair of nondegenerate hierarchies is compatible with $\distRF$, $\distCM$ and $\distCC$ in the sense of \refthm{thm.Stability} and \refprop{prop.Distance2Goal}, and can be efficiently computed with the same cost of $\distNav$ in $\bigO{n^2}$ time (\refapp{app.NavPathComputation}). 

%%%%%%%%%%%%%%%%%%%%%%%%%%%%%%%%%%%%%%%%%%%%%%%%%%%%%%%%%%%%%%%%%%%%%%%%%
%%%%%%%%%%%%%%%%%%%%%%%%%%%%%%%%%%%%%%%%%%%%%%%%%%%%%%%%%%%%%%%%%%%%%%%%%
\section*{Acknowledgements}
%%%%%%%%%%%%%%%%%%%%%%%%%%%%%%%%%%%%%%%%%%%%%%%%%%%%%%%%%%%%%%%%%%%%%%%%%
%%%%%%%%%%%%%%%%%%%%%%%%%%%%%%%%%%%%%%%%%%%%%%%%%%%%%%%%%%%%%%%%%%%%%%%%%

This work was funded in part by the Air Force Office of Science Research under the MURI FA9550-10-1-0567.


%%%%%%%%%%%%%%%%%%%%%%%%%%%%%%%%%%%%%%%%%%%%%%%%%%%%%%%%
%%%%%%%%%%%%%%%%%%%%%%%%%%%%%%%%%%%%%%%%%%%%%%%%%%%%%%%%
\appendix
%%%%%%%%%%%%%%%%%%%%%%%%%%%%%%%%%%%%%%%%%%%%%%%%%%%%%%%%
%%%%%%%%%%%%%%%%%%%%%%%%%%%%%%%%%%%%%%%%%%%%%%%%%%%%%%%%

%%%%%%%%%%%%%%%%%%%%%%%%%%%%%%%%%%%%%%%%%%%%%%%%%%%%%%%%
%%%%%%%%%%%%%%%%%%%%%%%%%%%%%%%%%%%%%%%%%%%%%%%%%%%%%%%%
\section{Proofs}
\label{app.Proofs}
%%%%%%%%%%%%%%%%%%%%%%%%%%%%%%%%%%%%%%%%%%%%%%%%%%%%%%%%
%%%%%%%%%%%%%%%%%%%%%%%%%%%%%%%%%%%%%%%%%%%%%%%%%%%%%%%%

%%%%%%%%%%%%%%%%%%%%%%%%%%%%%%%%%%%%%%%%%%%%%%%%%%%%%%%%
%%%%%%%%%%%%%%%%%%%%%%%%%%%%%%%%%%%%%%%%%%%%%%%%%%%%%%%%
\subsection{Proof of \reflem{lem.NNITriple}}
\label{app.NNITriple}
%%%%%%%%%%%%%%%%%%%%%%%%%%%%%%%%%%%%%%%%%%%%%%%%%%%%%%%%
%%%%%%%%%%%%%%%%%%%%%%%%%%%%%%%%%%%%%%%%%%%%%%%%%%%%%%%%

\begin{proof} Sufficiency is directly evident from \refdef{def:NNIMove} because  the cluster sets of a pair of nondegenerate hierarchies differ exactly by one cluster if and only if they are NNI-adjacent. To verify necessity, let the move $\prl{\treeA,P}$, $P \in \grandchildset{\treeA}$ join $\treeA$ to $\treeB$, and $R= \complementLCL{P}{\treeA}$ and $Q= \grandparentCL{P,\treeA} \setminus \parentCL{P,\treeA}$. By \refdef{def:NNIMove}, $\crl{\parentCL{P,\treeA}} = \crl{P \cup R}=\cluster{\treeA}\setminus \cluster{\treeB}$ and  $\crl{\grandparentCL{P,\treeA} \setminus P} = \crl{R \cup Q}=\cluster{\treeB}\setminus \cluster{\treeA}$. Further, $\prl{P,R,Q}$ is the only ordered triple of common clusters of $\treeA$ and $\treeB$ with the property that $\crl{P \cup R} = \cluster{\treeA} \setminus \cluster{\treeB}$ and $\crl{R\cup Q} = \cluster{\treeB}\setminus \cluster{\treeA}$ since the cluster sets of any two NNI-adjacent hierarchies differ exactly by one element. 
\end{proof}

%%%%%%%%%%%%%%%%%%%%%%%%%%%%%%%%%%%%%%%%%%%%%%%%%%%%%%%%
%%%%%%%%%%%%%%%%%%%%%%%%%%%%%%%%%%%%%%%%%%%%%%%%%%%%%%%%
\subsection{Proof of \reflem{lem:BinTreeRestriction}}
\label{app.BinTreeRestriction}
%%%%%%%%%%%%%%%%%%%%%%%%%%%%%%%%%%%%%%%%%%%%%%%%%%%%%%%%
%%%%%%%%%%%%%%%%%%%%%%%%%%%%%%%%%%%%%%%%%%%%%%%%%%%%%%%%

\begin{proof} To observe that $\treeres{K}\prl{\bintreetopspace_{\indexset}} \supseteq \bintreetopspace_{K}$, consider any two nondegenerate trees $\treeA \in \bintreetopspace_{K}$ and $\treeC \in \bintreetopspace_{\indexset \setminus K}$, and let $\treeB \in \bintreetopspace_{\indexset}$ be the nondegenerate tree with cluster set $\cluster{\treeB} = \cluster{\treeA} \cup \crl{\indexset} \cup \cluster{\treeC}$. Note that $\childCL{\indexset, \treeB} = \crl{K,\indexset \setminus K}$. Hence, we have from \refrem{rem:ClusterConcatenation} that $\treeA = \treeres{K}\prl{\treeB}$. To prove that $\treeres{K}\prl{\bintreetopspace_{\indexset}} \subseteq \bintreetopspace_{K}$, let $\treeB \in \bintreetopspace_{\indexset}$ and $I \in \cluster{\treeB}$ with the property that $\card{I \cap K} \geq 2$. Note that $I \cap K$ is an interior cluster of $\treeressym{K}{\treeB}$. We shall show that the cluster $I \cap K  \in \cluster{\treeressym{K}{\treeB}}$ always admits a bipartition in $\treeressym{K}{\treeB}$. That is to say,  there exist a  cluster $A \in \cluster{\treeB}$ with children $\crl{A_L,A_R} = \childCL{A,\treeB}$ such that $A \cap K = I \cap K$ and $A_L \cap K \neq \varnothing$ and $A_R\cap K \neq \varnothing$. Hence, $\childCL{I \cap K, \treeressym{K}{\treeB}} = \crl{A_L \cap K, A_R \cap K}$. Now observe that either $I_L \cap K \neq \varnothing$ and $I_R \cap K \neq \varnothing $ for $\crl{I_L,I_R} = \childCL{I,\treeB}$, or there exists one and only one descendant $D \in \descendantCL{I,\treeB}$ with $\crl{D_L, D_R} = \childCL{D,\treeB}$ such that $I \cap K = D \cap K$ and $D_L \cap K \neq \varnothing$ and $D_R \cap K \neq \varnothing$. Thus, all the interior clusters of $\treeressym{K}{\treeB}$  have exactly two children, which completes the proof.
\end{proof} 

%%%%%%%%%%%%%%%%%%%%%%%%%%%%%%%%%%%%%%%%%%%%%%%%%%%%%%%%
%%%%%%%%%%%%%%%%%%%%%%%%%%%%%%%%%%%%%%%%%%%%%%%%%%%%%%%%
\subsection{Proof of \reflem{lem.treeUltrametric}}
\label{app.treeUltrametric}
%%%%%%%%%%%%%%%%%%%%%%%%%%%%%%%%%%%%%%%%%%%%%%%%%%%%%%%%
%%%%%%%%%%%%%%%%%%%%%%%%%%%%%%%%%%%%%%%%%%%%%%%%%%%%%%%%

\begin{proof} The proof of the sufficiency for being an ultrametric is as follows. 
Positive definiteness and symmetry of $\dist_{\tree}$ are  evident from  \refeqn{eq.treeUltrametric} and \reflem{lem.treeUltrametric}.\ref{it.heightMonotonicity}-\ref{it.heightIndiscernible}. 
To show the strong triangle inequality, let $i\neq j \neq k \in \indexset$ and  $I = \cancCL{i}{j}{\tree}$, and so $\dist_{\tree}\prl{i,j} = h_{\tree}\prl{I}$.
Accordingly, let $\{I_i, I_j\} \subseteq \childCL{I,\tree}$ with the property that $i \sqz{\in} I_i$ and $j \sqz{\in} I_j$. 

If $k \in I$, without loss of generality, let $k \in I_i$, and so $k \not \in I_j$. 
Then, using \refeqn{eq.treeUltrametric} and \reflem{lem.treeUltrametric}.\ref{it.heightMonotonicity}, one can  verify  that $\dist_{\tree}\prl{i,k} \leq h_{\tree}\prl{I_i} \leq h_{\tree}\prl{I}$ and $\dist_{\tree}\prl{j,k} = h_{\tree}\prl{I}$ because $\cancCL{i}{k}{\tree} \subseteq I_i$ and $\cancCL{j}{k}{\tree} = I$.
Also note that if neither $k  \in I_i$ nor $k \in I_j$ (but still $k \in I$), then $\dist_{\tree}\prl{i,k} = \dist_{\tree}\prl{j,k} = h_{\tree}\prl{I}$ since $\cancCL{i}{k}{\tree} = \cancCL{j}{k}{\tree}= I$.
Similarly, if $k \not \in I$, then  $\dist_{\tree}\prl{i,k} \geq h_{\tree}\prl{I}$ and $\dist_{\tree}\prl{j, k} \geq h_{\tree}\prl{I}$ because only some ancestors of $I$ in $\tree$ might contain all $i,j,k$. 
Therefore, overall, one always has  $\dist_{\tree}\prl{i,j} \leq \max\prl{\big. \dist_{\tree}\prl{i,k}, \dist_{\tree}\prl{k, j}}$, which completes the proof  of the sufficiency. 

Let us continue with the necessity for being an ultrametric. 
Note that \reflem{lem.treeUltrametric}.\ref{it.heightIndiscernible} directly follows from positive definiteness of $\dist_{\tree}$.  %and $h_{\tree}\prl{I} > 0$ for all $I \in \cluster{\tree}$ and $\card{I} > 1$. 
Let $I \in \cluster{\tree} \setminus \crl{\indexset}$ be any non-singleton cluster of $\tree$ and $i \neq j \in I$ with the property that $\cancCL{i}{j}{\tree} = I$. %$\dist_{\tree}\prl{i,j} = h_{\tree}\prl{I}$. 
For any $k \in \complementLCL{I}{\tree}$, we always have 
$\cancCL{i}{k}{\tree} = \cancCL{j}{k}{\tree} = \parentCL{I,\tree}$.
%$\dist_{\tree}\prl{i,k} = \dist_{\tree}\prl{j,k} = h_{\tree}\prl{P}$,
%where $P = \parentCL{I,\tree}$. 
Now, using the ultrametric inequality of $\dist_{\tree}$, one deduces \reflem{lem.treeUltrametric}.\ref{it.heightMonotonicity} from
\begin{equation}
h_{\tree}\prl{I} = \dist_{\tree}\prl{i,j} \leq \max\prl{\big.\dist_{\tree}\prl{i,k}, \dist_{\tree}\prl{j,k}} = h_{\tree}\prl{\big. \parentCL{I,\tree}} ,
\end{equation}    
which completes the proof.
\end{proof}

%% References


\bibliography{NNIWalk}






% \end{thebibliography}

\end{document}

%% file: package.tex
\usepackage{comment}
\usepackage[%
hyperfigures=true,%
breaklinks=false,%
colorlinks=true,%
citecolor=green,%
linkcolor=blue%
]{hyperref}

\usepackage{dsfont} % for indicator function
\usepackage{bbm} % for indicator function - Type 3
\usepackage[cmex10]{amsmath}
\usepackage{amsfonts}
\usepackage{amssymb}
\usepackage{amsthm}
\usepackage{makeidx}
\usepackage{pifont}
\usepackage{graphicx, subfigure, epsf, latexsym, psfrag}

\usepackage{bm}
\usepackage{graphics, array}
\usepackage{epsfig}

\usepackage{latexsym}
\usepackage[mathcal]{euscript}
\usepackage[vcentermath]{youngtab}
\usepackage{multirow}
\usepackage{color}
\usepackage{pstricks}
\usepackage{multirow}
\usepackage{multicol}
\usepackage{lipsum}
\usepackage{lineno}
\usepackage{enumitem}

%% file: header.tex
	\theoremstyle{plain}

	\newtheorem{theorem}{Theorem}
	\newtheorem{lemma}{Lemma}
	\newtheorem{proposition}{Proposition}
	\newtheorem{corollary}{Corollary}
	
	\newtheorem{problem}{Problem}

	\newtheorem{definition}{Definition}
	
	\newtheorem{remark}{Remark}

    \newcommand{\prl}[1] 		{\left(#1\right)}
    \newcommand{\brl}[1] 		{\left[#1\right]}
    \newcommand{\crl}[1] 		{\left\{ #1\right\}}
    
    \newcommand{\R}             {\mathbb{R}} % Reals
 	\newcommand{\N}         		{\mathbb{N}} % Natural Numbers
  	\newcommand{\card}[1]		{\left| #1 \right|} % cardinality
	\newcommand{\norm}[1] 		{\left \| #1 \right \|} 
    \newcommand{\absval}[1] 		{\left | #1 \right |}
    
    \newcommand{\tr}[1]         {{#1}^\mathrm{T}}
     
	\newcommand{\mat}[1] 		{\mathrm{\mathbf{#1}}}

    \newcommand{\refeqn}[1]			{(\ref{#1})}
    \newcommand{\reffig}[1]			{Figure \ref{#1}}
    \newcommand{\refsec}[1]			{Section \ref{#1}}
    \newcommand{\refapp}[1]			{\ref{#1}}
    
	\newcommand{\reftab}[1]			{Table \ref{#1}}
	\newcommand{\refthm}[1]			{Theorem \ref{#1}}
	\newcommand{\refprop}[1]		    {Proposition \ref{#1}}
	\newcommand{\refcor}[1]			{Corollary \ref{#1}}
	\newcommand{\reflem}[1]			{Lemma \ref{#1}}
	\newcommand{\refdef}[1]			{Definition \ref{#1}}
	
	\newcommand{\refrem}[1]			{Remark \ref{#1}}
	\newcommand{\refprob}[1]		{Problem \ref{#1}}

	\newcommand{\indicator}          {\mathds{1}}%{\mathbbm{1}} %{\mathsf{1}\prl{#1}}

    \newcommand{\ldf}				{  \: {\mathbf := }  \: }

    \newcommand{\sqz}[1]            {\! #1 \!}

   \def \logicor {\vee}
   
   \newcommand{\THEN}{\;\Rightarrow\;}
   
   \newcommand{\minus}{\smallsetminus}

\newcommand{\bigO}[1] 			{\mathrm{O}(#1)}

%% file: notation.tex
\newcommand{\myTHEN}{{\;\Longrightarrow\;}}

\newcommand{\treenorm}[2]{{\left\Vert #2\right\Vert_{#1}}}
\newcommand{\crossnorm}[2]{{\left\Vert #2\right\Vert_{#1}}}

\newcommand{\projpaths}[2]{{\mathbf{Path}_{#1}\!\left(#2\right)}}
\newcommand{\projector}[1]{{\Gamma_\indexset\!\left(#1\right)}}
\newcommand{\res}[1]{\big|_{#1}}
\newcommand{\length}[1]{\ell\left(~#1~\right)}
\newcommand{\xindex}[2]{\left[\!\left[ #1\,|\,#2\right]\!\right]}

\newcommand{\dist} 				{\mathit{d}}  % generic dissimilarity measure
\newcommand{\distNNI}          	{\dist_{NNI}} % Nearest Neighbor Interchange (NNI) Distance
\newcommand{\distRF}           	{\dist_{RF}}  % Robinson-Foulds (symmetric difference) Distance
\newcommand{\distMS}             {\dist_{MS}}  % Matching Split Distance 
\newcommand{\distCM}             {\dist_{CM}}  % Crossing Distance
\newcommand{\distCC}             {\dist_{CC}}  % Cluster-Cardinality Distance
\newcommand{\distNav}            {\dist_{nav}} % NNI Navigation Distance 

\newcommand{\indexset}			{S} % Leaf Set
\newcommand{\treetopspace}      	{\mathcal{T}} % Tree Space 
\newcommand{\bintreetopspace}   	{\mathcal{BT}} % Binary Tree Space
\newcommand{\BT}[1]{\bintreetopspace_{#1}}
\newcommand{\tree}        		{\tau}
\newcommand{\treeA}             	{\sigma}
\newcommand{\treeB}             	{\tau}
\newcommand{\treeC}             	{\gamma}
\newcommand{\cluster}[1] 		{\mathcal{C}\prl{#1}} % Cluster Set
\newcommand{\intcluster}[1]     	{\mathcal{C}_{int} \prl{#1} } % Interior Cluster Set
\newcommand{\PowerSet}[1]       	{\mathcal{P}\prl{#1}} % Power Set
\newcommand{\diam}[1] 			{\mathrm{diam}\prl{#1}} % Diameter
\newcommand{\treeres}[1]        {res_{#1}}
\newcommand{\treeressym}[2]     {{#2}\big|_{#1}}

\newcommand{\myproj}[1]{{\mathbf{P}_{\indexset}\!\left(#1\right)}}

\newcommand{\compatible}  		{\bowtie} % Cluster Compatility
\newcommand{\CMat}              	{\mat{C}} % Compatibility Matrix
\newcommand{\XMat}              	{\mat{X}} % Crossing Matrix
\newcommand{\SMat}              	{\mat{S}} % Special Crossing Matrix

\newcommand{\parentCL}[1]			{\mathrm{Pr} \left( #1 \right)}
\newcommand{\grandparentCL}[1]		{\mathrm{Pr}^2\prl{#1}}
\newcommand{\grandchildset}[1]  		{\mathcal{G}\prl{#1}}
\newcommand{\childCL}[1]		    		{\mathrm{Ch}\prl{#1}}

\newcommand{\ancestorCL}[1]     		{\mathrm{Anc} \left( #1\right )}
\newcommand{\cancCL}[3]          	{\prl{#1  {\wedge}  #2}_{#3}}
\newcommand{\descendantCL}[1]   		{\mathrm{Des} \left( #1\right )}
\newcommand{\complementCL}[1]   		{{#1}^{C}}
\newcommand{\complementLCL}[2]      	{{#1}^{-#2}}
\newcommand{\depth}[1]          {\ell_{#1}}

\newcommand{\NNI}              	{\text{NNI}}
\newcommand{\NNIgraph}      	   	{\mathcal{N}}
\newcommand{\NNIdir}[1]{\NNIgraph_{#1}}
\newcommand{\NNIedgeset}       	{\mathcal{E}}
\newcommand{\NNIDedgeset}        {\widetilde{\mathcal{E}}}

\newcommand{\CommonXSplits}[1]  	{\mathcal{K}\prl{#1}} %clusters with crossing splits
\newcommand{\CommonXCL}         	{K} % a common cluster with a different split
\newcommand{\CLXSplitSet}[1]    	{\mathcal{I}\prl{#1}}
\newcommand{\CLXSplit}          	{I}
\newcommand{\DeepXSplit}[1]     	{\mathcal{R}\prl{#1}}
 